\documentclass[acmsmall]{acmart}

\AtBeginDocument{%
  }

\setcopyright{acmlicensed}
\copyrightyear{2025}
\acmDOI{10.1145/3727117}

\acmJournal{POMACS}
\acmYear{2025} \acmVolume{9} \acmNumber{2}
\acmArticle{25} \acmMonth{6}

\received{January 2025}
\received[revised]{April 2025}
\received[accepted]{April 2025}


\usepackage{tikz,bm}
\usepackage{thmtools, thm-restate}
\usepackage{array,adjustbox}
\usepackage{caption}
\usepackage{subcaption}
\usepackage{wrapfig}
\usepackage{enumitem}
\newcolumntype{C}[1]{>{\centering\let\newline\\\arraybackslash\hspace{0pt}}m{#1}}
\usetikzlibrary{shapes,arrows,positioning,automata}
\newcommand{\modstate}[2]{u_{#2,#1}^p}

\newcommand{\newedits}[1]{{#1}}
\newcommand{\E}{\mathbb{E}}
\newcommand{\norm}[1]{\left\lVert #1 \right\rVert}
\newcommand{\schedule}{{\bf u}}
\newcommand{\msrPolicy}{p}
\newcommand{\nmsrPolicy}{q}
\newcommand{\smsrPolicy}{r}
\newcommand{\msronePolicy}{p\mbox{-}1}
\newcommand{\mpath}{\omega}
\newcommand{\bP}{{\bf P}}
\newcommand{\wc}{\nu}
\newcommand{\msrstate}[2]{u_{#2,#1}^{\msronePolicy}}
\newcommand{\bgraph}[1]{\noindent\textbf{#1}}

\newtheorem*{theorem*}{Theorem}
\newtheorem*{lemma*}{Lemma}
\newtheorem{theorem}{Theorem}
\newtheorem{lemma}[theorem]{Lemma}

\newtheorem{remark}[theorem]{Remark}
\newtheorem{definition}{Definition}

\declaretheorem[numbered=yes,
name=Theorem]{throughputoptimal}

\begin{document}
\title[Improving Multiresource Job Scheduling with Markovian Service Rate Policies]{Improving Multiresource Job Scheduling with Markovian Service Rate Policies}

\author{Zhongrui Chen}
\email{jcpwfloi@cs.unc.edu}
\orcid{0009-0004-8694-4404}
\affiliation{%
 \institution{University of North Carolina at Chapel Hill}
 \streetaddress{201 S Colubmia St}
 \city{Chapel Hill}
 \state{North Carolina}
 \country{USA}
 \postcode{27599}
}
\author{Isaac Grosof}
\email{izzy.grosof@northwestern.edu}
\orcid{0000-0001-6205-8652}
\affiliation{%
 \institution{Northwestern University}
 \streetaddress{2145 Sheridan Rd, Evanston, IL}
 \city{Evanston}
 \state{Illinois}
 \country{USA}
 \postcode{60208}
}
\author{Benjamin Berg}
\email{ben@cs.unc.edu}
\orcid{0000-0002-4147-6860}
\affiliation{%
 \institution{University of North Carolina at Chapel Hill}
 \streetaddress{201 S Colubmia St}
 \city{Chapel Hill}
 \state{North Carolina}
 \country{USA}
 \postcode{27599}
}
\renewcommand{\shortauthors}{Zhongrui Chen, Isaac Grosof, and Benjamin Berg}




\begin{abstract}
Modern cloud computing workloads are composed of \emph{multiresource jobs} that require a variety of computational resources in order to run, such as CPU cores, memory, disk space, or hardware accelerators.
A single cloud server can typically run many multiresource jobs in parallel, but only if the server has sufficient resources to satisfy the demands of every job.
A scheduling policy must therefore select sets of multiresource jobs to run in parallel in order to minimize the \emph{mean response time} across jobs --- the average time from when a job arrives to the system until it is completed.
Unfortunately, achieving low response times by selecting sets of jobs that fully utilize the available server resources has proven to be a difficult problem.

In this paper, we develop and analyze a new class of policies for scheduling multiresource jobs, called Markovian Service Rate (MSR) policies.
While prior scheduling policies for multiresource jobs are either highly complex to analyze or hard to implement, our MSR policies are simple to implement and are amenable to response time analysis.
We show that the class of MSR policies is \emph{throughput-optimal} in that we can use an MSR policy to stabilize the system whenever it is possible to do so.
We also derive bounds on the mean response time under an MSR algorithm that are tight up to an additive constant.
These bounds can be applied to systems with different preemption behaviors, such as fully preemptive systems, non-preemptive systems, and systems that allow preemption with setup times.
We show how our theoretical results can be used to select a good MSR policy as a function of the system arrival rates, job service requirements, the server's resource capacities, and the resource demands of the jobs.
\end{abstract}

\maketitle
\section{Introduction}
\label{sec:intro}
Modern data center servers include a wide variety of computational resources, including CPU cores, network cards, local or disaggregated memory and storage, and specialized accelerators such as GPUs.
Each job request in the data center requires some set of diverse resources in order to run, such as a few CPU cores, some amount of memory, and some amount of storage \cite{maguluri2013scheduling,Maguluri2014HeavyClusters,tirmazi2020borg}.
Each job also has an associated \emph{service requirement}, which describes how long the job must run once it has been allocated its requested resources.
Because a single job requires several types of resources, we refer to these jobs as \emph{multiresource jobs}.
To serve multiresource jobs quickly, a single server generally aims to run many multiresource jobs in parallel.
However, the server also has a limited capacity for each resource.
Hence, a set of multiresource jobs can only be run in parallel if their aggregate demand for each resource is less than the server's capacity for that resource.
For example, in order to run some set of jobs concurrently, the number of total CPU cores demanded by the jobs should not exceed the total number of CPU cores in the server, and the total memory required by these jobs should not exceed the DRAM capacity of the server.
This raises the question of how a scheduling policy should choose sets of multiresource jobs to run in parallel on a data center server.

When scheduling multiresource jobs, the goal is often to reduce job \emph{response times}, the time from when a job arrives to the system until it is completed \cite{tirmazi2020borg,delimitrou2014quasar}.
In particular, this paper aims to design scheduling policies to minimize \emph{mean response time}, the long-run average response time across a stream of arriving multiresource jobs.
Unfortunately, optimizing the response times of multiresource jobs has proven to be difficult \cite{harchol2022multiserver, Ghaderi2016RandomizedCloud,maguluri2013scheduling,Stolyar2004MaxWeightTraffic}.
Much of the literature on multiresource jobs only proves system stability results \cite{Maguluri2014HeavyClusters,Ghaderi2016RandomizedCloud}, while response time analysis has been largely restricted to simple policies such as \emph{First-Come-First-Served} (FCFS) that perform poorly in practice \cite{grosof2023reset}.
Furthermore, many existing results on scheduling multiresource jobs suggest using a MaxWeight policy \cite{Maguluri2014HeavyClusters} that is too complex to run in a real system.
As a result, modern data centers are generally greatly overprovisioned to achieve low response times \cite{tirmazi2020borg}, and rely on heuristic scheduling policies, such as BackFilling policies like First-Fit, that admit no theoretical analysis or guarantees \cite{Grosof2022WCFS:Systems}.

This paper introduces a new class of policies for scheduling multiresource jobs called Markovian Service Rate (MSR) policies.
We prove strong stability results for this class of policies and bounds on mean response time.
Using these bounds, we derive an accurate approximation of the mean response time of an MSR policy.
Best of all, our MSR policies require a single offline optimization step, after which the policy runs in constant time per scheduling decision.

\subsection{The Problem: Scheduling to Reduce Wastage}
The key challenge in scheduling multiresource jobs is that it is difficult to maximize the average utilization of all the server resources. 
Intuitively, finding sets of jobs with high utilization of the server resources requires solving a bin-packing problem.
Hence, na\"ive scheduling policies can greatly underutilize system resources.
We describe the average amount of idle resources under a policy as the \emph{wastage} of the policy.
In general, high wastage leads to higher mean response time, and possibly even instability.

To illustrate this problem, we consider the FCFS policy for multiresource jobs described in \cite{grosof2023reset}.
Whenever a job arrives or departs the system, FCFS repeatedly examines the job at the front of the queue and puts the job in service as long as the server has enough of resources to accommodate it.
This process continues until either the queue is empty or the next job in the queue would exceed the server's capacity for one or more resources.
Although it is a simple, intuitive policy, FCFS can greatly underutilize the available system resources.
For example, consider the case where there are two \emph{types} of jobs, with type-1 jobs  requesting 4 CPU cores and type-2 jobs requesting 2 CPU cores.
Assume that each job must run for one second, the server has 8 CPU cores in total, and arriving jobs alternate between type-1 and type-2 jobs.
In this case, FCFS grabs a 4-core job and a 2-core job from the front of the queue.
The next job in the queue requires four cores, but only two are available, so the remaining two cores in the server will be idle.
Hence, in our example, FCFS always wastes at least two cores and completes at most two jobs per second.
Alternatively, consider a no-wastage policy that alternates between running four 2-core jobs together or two 4-core jobs together.
This no-wastage policy can complete up to three jobs per second.
As a result, the no-wastage policy may perform well in cases where the FCFS is unstable.
We formalize this argument in Section \ref{sec:algorithm}.

FCFS suffers in this example because its high wastage limits the number of jobs that can be run in parallel.
While it is easy to spot better packings of the jobs in this simple example, the complexity of finding low-wastage allocations explodes as the number of job types grows (e.g., some jobs request three cores) and as the number of resource types grows (e.g., jobs request CPUs, GPUs, and DRAM).

\subsection{Our Goal: Improved Low-complexity Scheduling Policies}
To avoid the pitfalls of an FCFS-style policy, the canonical work on scheduling multiresource jobs suggests finding low-wastage allocations by solving complex bin-packing problems repeatedly to continuously determine the next scheduler action \cite{Stolyar2004MaxWeightTraffic,Maguluri2014HeavyClusters}.
These are the so-called \emph{MaxWeight} algorithms that have been shown to maximize the capacity region of the system (i.e., throughput optimality) and achieve empirically low mean response time.
Unfortunately, running MaxWeight in real systems is impractical due to the complexity of the bin-packing problem that must be solved repeatedly.
MaxWeight also preempts jobs frequently, which may not be feasible in all systems.
It is a long-standing open problem to find and analyze policies that are both simple enough to run in practice, but still achieve low mean response time similar to a MaxWeight policy.

The closest the scheduling community has come to solving this problem is the Randomized-Timers approach of \cite{Ghaderi2016RandomizedCloud}, which describes a non-preemptive policy for scheduling multiresource jobs that can compute each scheduling decision with low complexity (complexity is linear in the number of job types).
Randomized-Timers is throughput-optimal, but the policy often results in high mean response time compared to MaxWeight.
Additionally, the Randomized-Timers algorithm is a complicated randomized algorithm and there is no analysis of its mean response time.

Our goal in this paper is to develop policies with the best features of MaxWeight and Randomized-Timers.
That is, we seek low-complexity policies that can achieve low mean response time under a variety of job preemption models.
Our policies should stabilize the system whenever possible, and should be accompanied by strong theoretical guarantees about mean response time.

\begin{table}
    \caption{Comparison of MSR policies to other multiresource job scheduling policies.}
    \begin{tabular}{|c|C{2cm}|c|C{2cm}|C{2.5cm}|}
    \hline
        \small{Policy} & \small{Throughput Optimal} & \small{Low Complexity} & \small{Mean Response Time Analysis} & \small{Preemption Model}\\
    \hline
         \small{FCFS} \cite{grosof2023reset} & & x & x &\small{non-preemptive}\\
         \hline
         \small{First-Fit \cite{Grosof2022WCFS:Systems}} & & x & &\small{non-preemptive}\\
         \hline
         \small{ServerFilling \cite{Grosof2022WCFS:Systems}\footnotemark[1]} & \small{Sometimes} & x & \small{Sometimes} &\small{preemptive}\\
         \hline
         \small{MaxWeight \cite{Maguluri2014HeavyClusters}} & x & & \small{Heavy traffic\footnotemark[2]} & \small{preemptive}\\
         \hline
         \small{Randomized-Timers \cite{Psychas2018RandomizedCloud}} & x & x & &\small{non-preemptive}\\
         \hline
         \textbf{\small{MSR}} & \textbf{x} & \textbf{x} & \textbf{x}& \textbf{\small{general}}\\
         \hline
    \end{tabular}
    \label{tab:comparison}
\end{table}
\footnotetext[1]{ServerFilling is throughput optimal and has a mean response time analysis in the restricted setting of single-dimensional jobs with power-of-two resource requirements, but is not throughput optimal in the full multiresource setting we consider.}
\footnotetext[2]{MaxWeight undergoes state space collapse, which can be used to give a heavy traffic analysis of mean response time \cite{Stolyar2004MaxWeightTraffic}.}

\subsection{Our Solution: MSR Scheduling Policies}
This paper analyzes the \emph{Markovian Service Rate} (MSR) class of policies for scheduling multiresource jobs.
MSR policies use a finite-state Continuous Time Markov Chain (CTMC) to determine which set of jobs to run at every moment in time.
The states of this CTMC, known as the \emph{candidate set} of possible scheduling actions,
as well as the CTMC's transition rates, are determined offline via a one-time optimization step that considers the service requirement distributions and arrival rates for a given workload.
We show how to select an MSR policy that stabilizes the system whenever it is possible to do so.
By deriving additively-tight mean response time bounds for any MSR policy, we also show that a well-chosen MSR policy can run with low complexity and achieve a low mean response time.
Specifically, we show how to construct an MSR policy that is constant-competitive with MaxWeight in a heavy-traffic limit.
Furthermore, we find that MSR policies are comparable to MaxWeight at all loads in simulation.

Our analysis generalizes to a variety of job preemption models.
We consider the cases where jobs are fully preemptible, where jobs are non-preemptible, and where job preemptions are accompanied by setup times.
In each case, we show how to select an appropriate MSR policy and analyze the mean response time under this policy.
Specifically, the contributions of the paper are as follows:
\begin{itemize}[leftmargin=2em]
\item First, Section \ref{sec:algorithm} introduces the class of MSR scheduling policies, which are easy to implement with low computational complexity.
We show that, under a variety of job preemption behaviors, there exists an MSR policy that can stabilize the system if it is possible to do so.
Hence, we say that the class of MSR policies is throughput-optimal under each preemption behavior.
\item Next, in Section \ref{sec:rt}, we analyze the mean response time of the multiresource job system under any MSR policy by decoupling the system into a set of $M/M/1$ Markovian Modulated Service Rate systems. 
By analyzing these decoupled problems, we prove additively tight bounds on the mean queue length of an MSR policy.
We use our bounds to derive a mean response time approximation for MSR policies that is shown to be highly accurate in simulation.
\item In Sections \ref{sec:preemptions} and \ref{sec:eval-borg}, we use our theoretical results to optimize the choice of MSR policy under a variety of job preemption behaviors. 
We validate our choices by comparing MSR policies to policies from the literature in both theory and in simulation using a trace from the Google Borg cluster scheduler.
We find that our MSR policies are competitive with these other policies,
while also being analytically tractable and simple to implement.
Specifically, after performing an offline optimization step, the MSR policies we evaluate run in constant time.
\end{itemize}

Table \ref{tab:comparison} compares our new MSR policies to existing policies for scheduling multiresource jobs.

\section{Prior Work}
\label{sec:prior}

\bgraph{Scheduling in Real-world Systems.}
Multiresource jobs are ubiquitous in datacenters and supercomputing centers where each job requires some amount of several different resources (e.g., CPUs, GPUs, memory, disk space, network bandwidth) in order to run.
As a result, the systems community has extensively considered the problem of scheduling multiresource jobs \cite{Borg:LargeScale,tirmazi2020borg,delimitrou2014quasar, lipari2012slurm, slurm, hindman2011mesos}.
Unfortunately, almost all of this work depends on a complex set of heuristic scheduling policies that do not provide theoretical performance guarantees.


\bgraph{The Multiserver Job Model.}
Prior theoretical work has considered the multiserver job model \cite{harchol2022multiserver} where each job demands multiple servers in order to run.
Several papers have developed throughput-optimal policies for this multiserver job system, such as the MaxWeight Policy \cite{maguluri2013scheduling} and the Randomized-Timers policy \cite{Ghaderi2016RandomizedCloud}.
However, neither of these policies admits a general analysis of mean response time.
Other work  \cite{Grosof2020StabilitySystems,WangWeina2021ZeroJobs,Grosof2022WCFS:Systems} has analyzed the multiserver job system in a variety of special cases such as restricting jobs to belong to one of two classes \cite{Grosof2020StabilitySystems}, considering the system in a variety of scaling limits \cite{WangWeina2021ZeroJobs}, and requiring job server demands to be powers-of-2 \cite{Grosof2022WCFS:Systems}. 
In these special cases, the analysis relies heavily on the fact that there is only one resource type, servers.

\bgraph{The Multiresource Job Model.}
The multiresource job model is a generalization of the multiserver job model.
In this setting, jobs can now request several different types of resources.
Several papers \cite{Psychas2018RandomizedCloud,Ghaderi2016RandomizedCloud, Im2015TightScheduling,Ghaderi2016SimpleCloud,maguluri2013scheduling, chen2024simple}  have considered the multiresource job model.
In \cite{maguluri2013scheduling} the MaxWeight policy is also shown to be throughput optimal in the multiresource job model. 
Similarly, the Randomized-Timers policy \cite{Psychas2018RandomizedCloud,Ghaderi2016RandomizedCloud} is throughput optimal in the multiresource job setting.
Analyses of mean response time in the multiresource job setting are more limited.
For example, \cite{grosof2023reset} recently analyzed the mean response time of a simple FCFS policy in the multiresource job case.
\cite{chen2024simple} proposes a throughput-optimal policy but the response time analysis assumes there are only 2 types of jobs.

\bgraph{Markov Modulated Service Rate Systems.}
Our plan is to analyze MSR policies in the multiresource job setting.
These are policies whose modulating processes are CTMCs.
As a result, the multiresource job systems that we analyze will look like a collection of single-server queueing systems with Markov modulated service rates.
These Markov modulated systems have been studied extensively in the queueing literature \cite{neuts1978m,regterschot1986queue,ciucu2018two,dimitrov2011single,eryilmaz2012asymptotically,grosof2023reset,burman1986asymptotic,falin1999heavy}.
While \cite{burman1986asymptotic,falin1999heavy} first studied the heavy-traffic behavior of a system with Markov modulated \emph{arrival rates},  \cite{dimitrov2011single} generalized their results to allow Markov modulated service rates.
Matrix analytic methods \cite{neuts1978m,ciucu2018two,regterschot1986queue} have also been used to analyze these systems.
More recently, \cite{eryilmaz2012asymptotically,grosof2023reset} have obtained response time bounds in Markov modulated systems by using drift-based techniques.
We will extend these drift-based techniques to help us analyze mean response time in the multiresource job model.

\bgraph{Scheduling in Networks.} While we focus on analyzing multiresource job scheduling, some work on scheduling in networks considers related problems.
Achievable throughput analysis in generalized switches is analogous to trying to find the stability region of multiresource job systems \cite{tassiulas1990stability,markakis2013max,Stolyar2004MaxWeightTraffic}.
Other work considers scheduling in wireless networks, where only a subset of wireless channels can send messages at the same time, leading to bin-packing effects \cite{Ni2012Q-CSMA:Networks,li2011distributed}.
Unfortunately, the policies discussed in this work are too complex for response time analysis \cite{Stolyar2004MaxWeightTraffic,Ni2012Q-CSMA:Networks,li2011distributed}, and/or are computationally expensive \cite{Stolyar2004MaxWeightTraffic}.

\newcommand{\Capacity}{{\bf P}}
\newcommand{\capacity}{P}
\newcommand{\Completions}{{\bf C}}
\newcommand{\completions}{C}
\newcommand{\MP}{{\bf u}}

\section{Model}
\label{sec:model}
Our goal is to model the performance of jobs running in a data center.
When a job is submitted to the system, a load balancer will first select a server to run each job, and each individual server is then responsible for scheduling its assigned jobs \cite{tirmazi2020borg}.
This paper models the behavior of a single server in this setting, deferring the load balancing problem to future work.

We consider a server with $R$ different types of resources.
Let the vector $\Capacity=(\capacity_1, \capacity_2, \cdots, \capacity_R)\in\mathbb{R}^R$ denote the resource capacity of the server, where $\capacity_j$ is the amount of resource $j$ that the server has.
The server is tasked with completing a stream of jobs that arrive to the system over time.
All dispatched jobs are either in service, or are held in a central queue that can be served in any order.

\subsection{Multiresource Jobs}
We refer to the jobs in our model as \emph{multiresource jobs}.
A multiresource job is defined by an ordered pair $({\bf d}, x)$, where ${\bf d}\in\mathbb{R}_+^R$ is a \emph{resource demand vector} denoting the job's requirement for each resource type and $x\in\mathbb{R}_+$ is the job's \emph{service requirement}, denoting the amount of time the job must be run on the server before it is completed.
We assume a job's resource demands are constant throughout its lifetime.
The server can therefore run multiple multiresource jobs in parallel if the sum of the jobs' resource demands does not exceed the server capacity for each resource.

We consider workloads composed of a finite number of \emph{job types} representing the different kinds of applications.
Intuitively, jobs of the same type represent different instances of the same program or different virtual machines of the same instance type.
Hence, jobs of the same type have the same resource demand vector.
Specifically, we assume there are $K$ types of jobs, and that all type-$i$ jobs share the same resource demand row vector ${\bf D}_i$.
Let ${\bf D}=({\bf D}_1, {\bf D}_2, \cdots, {\bf D}_K)$ denote the matrix of these demand vectors.
We assume the service requirements of type-$i$ jobs are sampled i.i.d. from a common service distribution, $S_i\sim\exp(\mu_i)$.
While the system knows the resource demands of each arriving job, we assume that service requirements are unknown to the system.

We define a \emph{possible schedule} to be a set of jobs that can be run in parallel without violating any of the resource constraints of the server.
We describe a possible schedule using a vector $\schedule=(u_1, u_2, \cdots, u_K)\in\mathbb{Z}_+^K$, where $u_i$ denotes how many type-$i$ jobs we are putting into service in parallel.
Because service requirements are exponentially distributed, this description suffices to describe the behavior of the system.
The vector $\schedule$ is a possible schedule if and only if
$\sum_{i=1}^K u_i{\bf D}_i= \schedule {\bf D}\le \Capacity.$
We call the set of all possible schedules the \emph{schedulable set}, $\mathcal{S}$.
Formally,
$\mathcal{S}=\{\schedule\in\mathbb{Z}_+^K\mid \schedule \cdot {\bf D}\le \Capacity\}.$

\subsection{Arrival Process}
We consider the case where a stream of incoming jobs arrives at the system over time.
Let $\{A_i(t), t\ge 0\}$ be the counting process tracking the number of type-$i$ arrivals before time $t$. 
For convenience, we define $A_i(t_1, t_2):= A_i(t_2)-A_i(t_1)$ as the number of arrivals during the time interval $[t_1, t_2)$.
We assume $A_i$ evolves according to a Poisson process with rate $\lambda_i$.
Let $\bm{\lambda}=(\lambda_1, \lambda_2, \cdots, \lambda_K)\in\mathbb{R}_+^K$ denote the vector of job arrival rates and $\Lambda=\sum_{i=1}^K\lambda_i$ denote the total arrival rate.
Finally, let ${\bf A}(t)=(A_1(t), A_2(t), \cdots, A_K(t))$ denote the vector of cumulative arrivals by time $t$.

\subsection{Scheduling Policies}
\label{sec:def-policy}
A \emph{scheduling policy} $p$ in the multiresource job model chooses a possible schedule to use at every moment in time to serve some subset of the jobs in the system.
We describe $p$ via a \emph{modulating process} $\{m^p(t), t\ge 0\}$ with $N^p$ states, where $m: \mathbb{R}_+ \rightarrow \{1, 2, \cdots, N^{\msrPolicy}\}$.
Note that $N^p$ is a property of the policy $p$, not an exponential function.
Let $m^{\msrPolicy}$ be the stationary distribution of the modulating process.
Each state in the modulating process corresponds to a schedule in $\mathcal{S}$.
Let $\schedule^{\msrPolicy}(t)=\left(u_1^{\msrPolicy}(t), u_2^{\msrPolicy}(t), \cdots, u_K^{\msrPolicy}(t)\right)$ be the schedule that corresponds to $m^{\msrPolicy}(t)$ for any $t\ge 0$. For notational convenience, we also define $u_{s, i}^{\msrPolicy}=u^{\msrPolicy}_{i}(t \mid m^p(t)=s)$ for all $s\in\{1,2,\cdots, N^{\msrPolicy}\}$ and $i\in\{1,2,\cdots, K\}$.
Note that the schedules corresponding to $m^{\msrPolicy}(t)$ may, in general, depend on the system state (e.g., queue lengths).

We define a class of simple scheduling policies called \emph{Markovian Service Rate} (MSR) policies in Definition \ref{def:MSR}.
\begin{definition}[MSR Policies]
    \label{def:MSR}
    An MSR policy $\msrPolicy$ is a scheduling policy whose modulating process $\{m^{\msrPolicy}(t), t\ge 0\}$ is an irreducible Continuous Time Markov Chain (CTMC) with $N^{p}<\infty $ states, where the updates of the CTMC are conditionally independent of the underlying queueing system given the system events (e.g. completions and setups).
    Let ${\bf G}^{p}$ denote the infinitesimal generator of the CTMC.
    Let $\schedule^p$ denote the limiting distribution of the schedule used by policy $p$, and let $\E[\schedule^{p}]$ be the steady-state average schedule under policy $\msrPolicy$.
\end{definition}

MSR policies are simple in the sense that they use a finite-state CTMC to choose their schedules at every moment in time.
Furthermore, this CTMC is required to be mostly independent of the system state, although its state transitions can correspond to arrivals and departures in the system.
Specifically, the MSR policy's modulating process is not allowed to be correlated with any queue length information.
Our definition is compatible with the definition of an MMSR system in \cite{grosof2023reset}.
In this paper, we will show that the class of MSR policies is throughput-optimal and can achieve low mean response time while also being simpler to analyze and implement than existing approaches such as MaxWeight and Randomized-Timers.

\subsection{Preemption}
As described in Section \ref{sec:prior}, the choice of scheduling policy for multiresource jobs has historically depended on the preemption behavior of jobs.
To capture this dynamic, we consider jobs with a variety of preemption behaviors.
Each preemption behavior imposes a different set of constraints on what a policy's modulating process can do.
For instance, when jobs are preemptible with no overhead, a modulating process may transition arbitrarily between schedules by preempting jobs as necessary.
However, when jobs are non-preemptible, a modulating process cannot perform transitions that would necessitate preemptions.
For example, a non-preemptive policy cannot directly transition to a schedule with a much smaller number of type-$i$ jobs in service than the current schedule.
This paper considers the case where jobs are \emph{fully preemptible} with no overhead, the case where jobs are \emph{non-preemptible}, and the case where jobs can be preempted by waiting for a \emph{setup time} \cite{gandhi2013exact}.


The class of MSR policies can handle multiresource jobs under each of these preemption models by imposing different restrictions on its modulating process.
For preemptible jobs, the modulating process can be any irreducible, finite-state CTMC.
For non-preemptible jobs, the modulating process must avoid preemptions altogether.
That is, jobs can be added to the schedule if there is sufficient capacity, otherwise jobs must be removed from the schedule one at a time at moments of job completions.
When jobs are preemptible with setup times, the modulating process will spend time in some additional states that correspond to the required setup times.
Details of each of these restrictions on state transitions are discussed formally in Section \ref{sec:algo:xput}.
We refer to MSR policies with \emph{preemptible} jobs, \emph{non-preemptible} jobs, and jobs with \emph{setup times} as \emph{pMSR}, \emph{nMSR}, and \emph{sMSR} policies, respectively.
\subsection{Queueing Dynamics}
\label{sec:queue-dynamics}
We model the queue length as the CTMC $\{{\bf Q}(t)=(Q_1(t), Q_2(t), \cdots, Q_K(t)), t\ge 0\}$, where $Q_i(t)$ tracks the number of type-$i$ jobs in system at time $t$.  The queueing dynamics of the system follow
\begin{equation}{\bf Q}(t+\delta)={\bf Q}(t)+{\bf A}(t, t+\delta)-\hat{\Completions}(t,t+\delta)\label{eq:dyn},\end{equation}
where $\hat{\Completions}(t, t+\delta)$ denotes the number of jobs that completed service during time interval $[t, t + \delta)$. 
We call $\hat{\Completions}(t, t+\delta)$ \emph{actual completions} during time interval $[t, t+\delta)$.

Then, we will relate $\hat{\Completions}(t, t+\delta)$ to the schedule we pick at time $t$, ${\bf u}(t)$.
Note that the server may not be serving $\schedule_i(t)$ type-$i$ jobs at time $t$ if there are not enough type-$i$ jobs in the system.
The actual number of type-$i$ jobs in service is $\min(\schedule_i(t),{\bf Q}_i(t))$.
We let $\Completions(t, t+\delta)$ denote the number of \emph{potential completions} during the interval $[t, t+\delta)$ by simulating the number of jobs that would have completed if exactly ${\bf u}(t)$ jobs were in service.
Consequently, $\Completions(t,t+\delta)\ge \hat{\Completions}(t,t+\delta)$.
We let ${\bf Z}(t, t+\delta)$ denote the \emph{unused service} during time interval $[t, t+\delta)$ where $\hat{\Completions}(t,t+\delta)+{\bf Z}(t, t+\delta)={\Completions}(t,t+\delta)$.

We can then rewrite \eqref{eq:dyn} as
${\bf Q}(t+\delta)={\bf Q}(t)+{\bf A}(t, t+\delta)-\Completions(t, t+\delta)+{\bf Z}(t, t+\delta)$, or
$${\bf Q}(t+\delta)=({\bf Q}(t)+{\bf A}(t, t+\delta)-\Completions(t, t+\delta))^+,$$
where the positive part of $x$, $(x)^+$, is defined as $(x)^+=\max(x, 0)$.
Note that for any job type, $i$, ${\bf Z}_i(t,t+\delta)>0$ implies ${\bf Q}_i(t+\delta)=0$ because the queue has to be empty before unused service occurs.
Because job service requirements are exponentially distributed as $\exp(\mu_i)$, we have
$$\\E[\hat{\Completions}(t,t+\delta)] \leq E[{\Completions}(t,t+\delta)]=\bm{\mu}\otimes \int_{t}^{t+\delta} \schedule(x)dx \qquad \forall \delta > 0,$$
where $\bm{\mu}$ is the \emph{service requirement vector} defined as $\bm{\mu}=(\mu_1, \mu_2, \cdots, \mu_K)$, \newedits{and $\otimes$ denotes the \textit{Hadamard product}}.

\subsection{Performance Metrics}
\label{sec:capacity-region}
\emph{Our goal is to analyze the steady-state behavior of the server as $t \rightarrow \infty$.}
We define
$${\bf Q} \sim \lim_{t\rightarrow\infty} {\bf Q}(t).$$
We are interested in analyzing the \emph{steady-state mean queue length} of the server, $\E[{Q}]$, where
$$\E[Q]=\frac{1}{\Lambda}\langle \bm{\lambda}, \E[{\bf Q}]\rangle=\frac{1}{\Lambda}\langle \bm{\lambda}, \lim_{t\rightarrow\infty} \E[{\bf Q}(t)]\rangle.$$
We can then apply Little's Law to determine the mean response time across jobs, $\E[T]$, as 
$$\E[T]=\frac{\E[Q]}{\Lambda},$$
where $T$ is the \emph{steady-state response time of a job in the system}.
Hence, to analyze $\E[T]$, we will focus on analyzing $\E[Q]$.

While there has been significant work on analyzing the fairness of a scheduling policy by looking at higher moments of response time \cite{im2020dynamic} or  \emph{mean slowdown} \cite{wierman2003classifying}, this work focuses on optimizing the mean response time across jobs.
As we will discuss in Section \ref{sec:pmsr}, we can design policies that both achieve good mean response time and follow mild fairness properties.


Note that the above steady-state limits may not exist in all cases.
We say the system is \emph{stable} under a particular scheduling policy if ${\bf Q}(t)$ is ergodic and the steady-state queue length distribution exists \cite{harchol2013performance}.
When the system is stable, the steady-state distribution of $Q$ exists and $\E[{\bf Q}]<\infty$.
The \emph{capacity region}, $\mathcal{C}^{p}$, of a scheduling policy $p$ is defined as the set of arrival rates such that the system is stable under $p$.
The capacity region of a \emph{system}, $\mathcal{C}$, is the set of all arrival rates such that \emph{some} scheduling policy exists that stabilizes the system.
If a policy's capacity region equals the capacity region of the system, then we say this scheduling policy is \emph{throughput-optimal} \cite{Maguluri2014HeavyClusters}.

The results of \cite{Psychas2018RandomizedCloud,Maguluri2014HeavyClusters} show that the capacity region of any multiresource job system is
\begin{equation}\mathcal{C}=\{\bm{\lambda} \mid (1+\epsilon)\bm{\lambda} \oslash \bm{\mu} \in Conv(\mathcal{S}), \epsilon >0\}\label{eq:capacity}\end{equation}
where $Conv(\mathcal{S})$ is the convex set of points that can be written as a weighted average of the points in $\mathcal{S}$ and $\bm {\lambda}\oslash \bm{\mu}$ denotes Hadamard division between the vectors.
Intuitively, there exists a scheduling policy that can stabilize the system if and only if there exists a convex combination of possible schedules such that the average service rate of type-$i$ jobs exceeds the average arrival rate of type-$i$ jobs for all $i\in[1, K]$.

Here, we define the load on the system in terms of utilization of the server when the arrival rate vector is $\bm{\lambda}$.
Specifically, we define the \emph{system load}, $\rho>0$, to be the smallest number such that 
\begin{equation}
    \frac{\bm{\lambda}\oslash \bm{\mu}}{\rho} \in Conv(\mathcal{S})\label{eq:system-load}.
\end{equation}
Note that if $\bm{\lambda}\in\mathcal{C}$, then $\rho<1$ according to \eqref{eq:capacity}.


Then, \cite{Maguluri2014HeavyClusters} defines the \emph{MaxWeight} policy $\schedule^{\text{MaxWeight}}(t)$ as the solution to the optimization problem
$$\schedule^{\text{MaxWeight}}(t)=\arg\max_{\schedule\in\mathcal{S}} \langle \schedule, {\bf Q}(t)\rangle.$$
This policy is known to be throughput-optimal.
Unfortunately, MaxWeight also requires frequent preemption and repeatedly optimizing over the set of all possible schedules.
Solving this NP-hard optimization problem can be prohibitively expensive.
We aim to devise and analyze throughput-optimal policies that avoid this complexity.

\subsection{Summary of Modeling Assumptions}
\label{sec:assumptions}
To maintain tractability, our model includes a few assumptions that warrant discussion.

\bgraph{Job Sizes and Arrivals.}  We assume that job sizes are exponentially distributed, while real-world job sizes generally follow heavily-tailed distributions.
Fortunately, the analytic techniques in Section \ref{sec:rt} can potentially be extended to handle phase-type job sizes, providing a promising direction for future work.
We describe these extensions in Appendix \ref{sec:appendix:phase-type}.
We have also assumed a stationary Poisson arrival process that is less variable than arrivals in a real system.
We address both of these issues by evaluating our results using real-world data from the Google Borg scheduler in Section \ref{sec:eval-borg}.

\bgraph{Steady State Analysis.}  Our analysis develops queueing-theoretic results about the steady-state behavior of the system.
In a system where the workload changes rapidly and unpredictably, one might be concerned that the system never reaches stationarity.
The results of Section \ref{sec:eval-borg} show that our scheduling policies still perform well over the course of a thirty-day trace of real-world datacenter traffic.
Our analysis should be viewed as complementary to the worst-case and competitive analysis approaches that assume adversarial workloads and result in loose performance bounds \cite{leonardi2007approximating}.

\bgraph{Resource Requirements.} 
We assume that job resource requirements are constant throughout a job's lifetime.
While jobs in some systems can change resource demands in real time \cite{berg2018}, systems like Borg let users reserve resources statically.
While jobs may not fully use their demanded resources, scheduling decisions in reservation-based systems are still largely based on these reservations \cite{tirmazi2020borg}.

We also assume that the jobs' resource requirements are known beforehand.
This assumption is practical in some settings, such as the case where a cloud provider sells virtual machines of different instance types.
In other cases, the demands may have to be predicted by looking at recent runs.

\bgraph{Heterogeneous Resources.}
While the fleet of servers in a modern datacenter is heterogeneous, this work focuses on scheduling within a single server.
Within a server, we assume that resources of the same type are homogeneous.
For some resources, such as GPUs, the performance difference between GPU models can be so huge that they should be handled as separate resource types.
For CPU cores, differences between cores are more easily abstracted away from the scheduler.
For example, Google Borg normalizes CPU requests to \emph{Google Compute Units} so that 1 unit provides the same performance regardless of the underlying hardware \cite{tirmazi2020borg}.
While resources such as memory and disk space can have variable performance characteristics, there are many mechanisms in software and hardware designed to make these resources appear homogeneous as a program runs.

\section{Overview of Results}
The rest of this paper is devoted to analyzing various MSR policies.
Broadly, our theoretical results can be grouped into two categories.
First, in Section \ref{sec:algorithm}, we prove results about the stability of the system under various MSR policies, and show how to select an MSR policy that stabilizes the system under various preemption models.
Second, in Sections \ref{sec:rt} and \ref{sec:preemptions}, we analyze the mean response time under an MSR policy and use this analysis to further guide the selection of a stable MSR policy with low mean response time.
We now provide an overview of these results.

\bgraph{Stability Results.}
First, we prove the throughput-optimality of the class of MSR policies when jobs are fully preemptible.
In Theorem \ref{thm:throughput-optimal}, we show that if it is possible for any policy to stabilize the system, there exists a pMSR policy that will stabilize the system.

\begin{restatable*}{theorem}{throughputoptimal}
    \label{thm:throughput-optimal}
    When scheduling preemptible jobs with no preemption overhead, the class of MSR policies is throughput-optimal. That is, if there exists a scheduling policy that can stabilize a multiresource job system with $K$ job types and arrival rates $\bm{\lambda}$, then there exists an MSR policy, $p$, with $N^p\le K$ candidate schedules that stabilizes the system. Specifically, an MSR policy can stabilize a system with preemptible jobs for any $\bm{\lambda}\in\mathcal{C}$.
\end{restatable*}

This result has a nice geometric intuition.
We prove that choosing a pMSR policy to stabilize the system is equivalent to finding a set of candidate schedules that lie on $Conv(\mathcal{S})$ whose completion rate vectors span the arrival vector in the $K$-dimensional space.
We then argue that, by Carathéodory theorem, there exists a set of at most $K$ schedules that satisfy this condition, meaning there exists an MSR policy, $p$, with $N^p \le K$.


We note that our proof of Theorem \ref{thm:throughput-optimal} is not constructive --- it does not provide a method for selecting a pMSR policy.
Hence, in Corollary \ref{cor:pmsr-construction}, we construct a pMSR policy to stabilize the system by solving a constrained optimization problem over the space of possible modulating processes.
We state this formally in Corollary \ref{cor:pmsr-construction}.

\begin{restatable*}{corollary}{constructionpmsr}
\label{cor:pmsr-construction}
For a pMSR policy $p$, let ${\bf W}^p$ denote the matrix of candidate schedules,
$${\bf W}^p=\begin{pmatrix}
    u^p_{1,1} & u^p_{1,2} &\cdots & u^p_{1, K}\\
    \cdots && \cdots \\
    u^p_{N^p,1} & u^p_{N^p,2} &\cdots & u^p_{N^p, K}
\end{pmatrix}\in\mathbb{Z}^{+}_{N^p\times K}.$$ 
Let $\bm{\pi}^p$ be the limiting distribution of the modulating process $m^p(t)$,
$\bm{\pi}^p=(\pi^p_1, \pi^p_2, \cdots, \pi^p_{N^p})\in\mathbb{R}^{N^p}$. 

To construct a stable pMSR policy, it is necessary to find values ${\bf W}^p$ and $\bm{\pi}^p$ such that
$$\begin{cases}
    {\bf W}^p{\bf D} \le {\bf P} \quad \textnormal{(Capacity constraint)}\\
    \pi^p_i \ge 0, \forall i \quad\textnormal{(Positive fractions)}\\
    \pi^p_1 + \pi^p_2 + \cdots + \pi^p_{N^p} \le 1 \quad\textnormal{(Fractions sum to 1)}\\
    \bm{\pi}^p{\bf W}^p > \bm{\lambda} \oslash \bm{\mu} \quad \textnormal{(Stability condition from Lemma \ref{thm:existence})}\\
    {\bf W}^p \textnormal{ integers}
\end{cases}.$$
Because Theorem \ref{thm:throughput-optimal} guarantees that there exists an MSR policy, $p$, with $N^p\le K$ whenever it is possible to stabilize the system, we can safely assume that $N^p=K$ when solving this optimization problem.

Given the solution to this optimization problem, we can  construct a modulating process with candidate schedules $\bm{W}^p$ and limiting distribution is ${\bm \pi}^p$ by solving a system of $N^p + 1$ linear equations.
\end{restatable*}

Both Theorem \ref{thm:throughput-optimal} and Corollary \ref{cor:pmsr-construction} make the strong assumption that all jobs can be preempted with no overhead at any time.
Fortunately, we can generalize both results to the case where jobs are non-preemptible and the case where preemptions incur setup times.
These generalizations are proven as Theorems \ref{thm:nmsr-existence} and \ref{thm:smsr-existence}, respectively.
Both theorems show how a pMSR policy constructed using Corollary \ref{cor:pmsr-construction} can be adjusted to construct a stable policy under different preemption models, and thus both of these proofs are constructive.

\bgraph{Response Time Analysis Results.} While Corollary \ref{cor:pmsr-construction} constructs a stable MSR policy, there may be many MSR policies that can stabilize the system.
Our goal is therefore to analyze the mean response time of a given MSR policy, and use this analysis to choose an MSR policy with low mean response time.

Our analysis exploits the fact that the MSR policy's modulating process is independent of the system state.
This allows us to decouple the system and individually analyze the number of type-$i$ jobs in the system for each of the $K$ job types.
Specifically, we analyze the expected number of type-$i$ jobs under an MSR policy $p$, $\E[Q^p_i]$, by comparing it to $\E[Q^{\msronePolicy}_i]$, the number of jobs in a single-server queue with Markov-modulated service rates.
We refer to the single-server system as the MSR-1 system.

We first note that the difference between the mean queue lengths of these systems is at most the maximum number of type-$i$ jobs served in parallel. We prove this formally in Theorem \ref{theorem:msr-difference} via a coupling argument.

\begin{restatable*}{theorem}{msrdifference}
    \label{theorem:msr-difference}
    For any job type, $i$, consider an MSR policy $\msrPolicy$ and its corresponding MSR-1 policy $\msronePolicy$. 
    Let $\beta_i^{\msrPolicy}$ denote the maximum number of type-$i$ jobs served in parallel under $\msrPolicy$.  Then,
    $$\E[Q_i^{\msronePolicy}]\le \E[Q^{\msrPolicy}_i] \le \E[Q_i^{\msronePolicy}]+\beta_i^{\msrPolicy}.$$
\end{restatable*}

Next, we leverage the results of \cite{grosof2023reset} to bound $\E[Q^{\msronePolicy}_i]$.
Our bound comprises two terms.
First, there is a primary term that scales as $\Theta\left(\frac{1}{1-\rho_i}\right)$, where $\rho_i$ is the load on the decoupled type-$i$ system.
Second, there is a term corresponding to unused service --- 
 times when the type-$i$ system has unused idle capacity.
We combine this bound with Theorem \ref{theorem:msr-difference} to prove a bound on $\E[Q^p_i]$ in Theorem~\ref{thm:msr-bound} (see Section \ref{sec:bounds} for a full theorem statement along with the required notation).
In addition to our upper and lower bounds on mean queue length, we develop a simple heuristic to approximate $\E[Q^p_i]$ that always falls between our bounds.

We use this analysis to guide our choice of MSR policies under various preemption models in Section \ref{sec:preemptions}.
Ideally, we could select a policy by directly optimizing our response time bounds or approximations subject to the quadratic stability constraints proven in 
Corollary \ref{cor:pmsr-construction}.
However, our bounds and approximations include terms that must be computed numerically once the modulating process is fixed, so it is intractable to optimize over these expressions directly.
We instead suggest using a simple heuristic to choose an MSR policy --- to minimize the maximum value of $\rho_i$ for any job type, $i$.
This heuristic provides several benefits:
\begin{itemize}
    \item The heuristic is linear.  Hence, we can formulate a Mixed Integer Quadratic Constrained Programming (MIQCP) problem that optimizes the heuristic subject to the quadratic constraints that guarantee system stability.  A solution to the MIQCP, as defined in equation \eqref{eq:optimization} of Section \ref{sec:pmsr}, can be used to construct a stable pMSR policy, $p$.
    \item The pMSR policy we select performs well under heavy load. 
    This holds because minimizing our heuristic also minimizes system load, $\rho$, and the primary term in our bound on $\E[Q^p]$ scales as $\Theta(\frac{1}{1-\rho})$.
    Notably, we prove in Theorem \ref{thm:ht} that in the heavy traffic limit as $\rho \to 1$, $p$ is constant-competitive with MaxWeight.  
    \item The heuristic is minimized by balancing the load experienced by each job type.
    Hence, while fairness is not a central concern of this paper, we show that an MSR policy can achieve low overall mean response time without arbitrarily hurting one type of jobs.
\end{itemize}

\bgraph{How to use our results in practice.}
Running a stable MSR policy with low mean response time consists of both an offline and an online step.

First, we initialize the policy by solving an MICQP, \eqref{eq:optimization}, to determine the candidate set and limiting distribution of the modulating process. 
Next, we solve a system of equations to construct a modulating process that meets this specification.
While there are many possible structures we can use for the modulating process, a simple heuristic is to arrange the candidate schedules in a loop structure.
The structure of the modulating process is discussed in Section \ref{sec:preemptions}.
While this step is computationally expensive, it happens offline.

Second, the online portion of the MSR policy runs with complexity at most $O(K)$ by maintaining at most $K$ exponential timers to track transitions of the modulating process.
When using a modulating process with a loop structure, as suggested above, only 1 exponential timer is needed and the online portion of the MSR policy runs with complexity $O(1)$.
Hence, MSR policies move the vast majority of their complexity offline, resulting in efficient online execution.

\section{Markovian Service Rate (MSR) Policies}
\label{sec:algorithm}

In this section, we will prove the class of MSR policies is \emph{throughput-optimal} for scheduling multiresource jobs.
In other words, if there exists a scheduling policy that can stabilize a given system, we can select an MSR policy that stabilizes the system.
The results of this section only provide stability conditions for various MSR policies.
In order to construct performant MSRs in practice, we will need to use the more detailed response time analysis developed in Section \ref{sec:rt}.
We will then show in Section \ref{sec:preemptions} how to use our stability conditions and response time analysis to construct performant MSR policies under a variety of preemption models.

\subsection{Throughput-Optimality of MSR policies}
\label{sec:algo:xput}
Our high-level goal in this section is to show that the class of MSR policies is throughput-optimal.
However, we will see that proving this claim will depend on the preemption behavior of jobs.
Hence, we will prove this style of result for a variety of preemption behaviors in Theorems \ref{thm:throughput-optimal}, \ref{thm:nmsr-existence}  and \ref{thm:smsr-existence}.

We begin by first establishing necessary and sufficient stability criteria for any MSR policy.
Intuitively, the system should be stable whenever the long-run average completion rate of an MSR policy exceeds the arrival rate for every job type.
We prove this stability condition as Lemma \ref{thm:existence}.

\begin{restatable}[Stability Condition]{lemma}{existence}
    \label{thm:existence}
    Given a multiresource system with $K$ job types, arrival vector $\bm{\lambda}$ and service requirement vector $\bm{\mu}$, the system is stable under an MSR policy $p$ if and only if the long-run average completion rate for each job type exceeds the long run arrival rate.
    That is, the system is stable if and only if $\bm{\mu}\otimes \E[\schedule^p]> \bm{\lambda}$, where $\otimes$ denotes the Hadamard product between vectors.
\end{restatable}
\begin{proof}
This proof uses the Foster-Lyapunov criteria \cite{Stolyar2004MaxWeightTraffic} to establish stability conditions for the system.
We use the L2-norm of the queue length vector ${\bf Q}(t)$ as our Lyapunov function and show that the required drift conditions are satisfied if and only if the average completion rate of jobs is greater than the average arrival rate.
See Appendix \ref{sec:stability} for details.
\end{proof}
Lemma \ref{thm:existence} applies to any MSR policy regardless of the preemption model being considered.
Hence, we can use these stability criteria to reason about throughput-optimality under a variety of preemption behaviors in the following sections.

\bgraph{Fully Preemptible Jobs.}
Next, we show that Lemma \ref{thm:existence} suffices to prove throughput-optimality when scheduling fully preemptible jobs.
In this case, we assume that jobs can be preempted at arbitrary times, arbitrarily frequently, with no associated overhead.
As such, we place no constraints on the modulating process.
We call an MSR policy for scheduling fully preemptible multiresource jobs a \emph{pMSR policy}.
We prove the throughput-optimality of pMSR policies in Theorem \ref{thm:throughput-optimal}.

\throughputoptimal
    \begin{proof}
        As described in Section \ref{sec:capacity-region}, the capacity region $\bm{\lambda}\in\mathcal{C}$ is known and is given by \eqref{eq:capacity}. 
        Equation \eqref{eq:capacity} tells us that $\bm{\lambda} \oslash \bm{\mu}$ must lie strictly on the interior of the convex hull of $\mathcal{S}$.
        Hence, for some $\epsilon'>0$, $(1+ \epsilon')\bm{\lambda} \oslash \bm{\mu}$ lies on the border of the convex hull of $\mathcal{S}$, and this point can be written as a convex combination of schedules on the border of the convex hull of $\mathcal{S}$.
        By Carathéodory's Theorem \cite{fabila2017caratheodory}, there exists at least one such convex combination consisting of at most $K$ schedules.
        Our goal is to set the transitions of the modulating process so that the stationary probabilities of being in each state match the weights of the convex combination.
        Finding appropriate transition rates then amounts to solving a system of $K-1$ linearly independent equations and $K(K-1)$ unknowns.
        This guarantees that we can construct a pMSR policy, $p$, with $N^p\leq K$ such that $\E[\schedule^p] = (1+\epsilon')\bm{\lambda} \oslash \bm{\mu}$ for some $\epsilon'>0$.
        By Lemma \ref{thm:existence}, $\{{\bf Q}^p(t), t\ge 0\}$ is positive recurrent and the system is stable.
\end{proof}

Crucially, Theorem \ref{thm:throughput-optimal} says that a stable MSR policy only needs $K$ states in its modulating process, confirming that pMSR policies provide a low-complexity solution for scheduling multiresource jobs.
At any moment in time, the only information the policy needs to maintain in order to stabilize the system are $K$ states comprising the candidate set, the $K(K-1)$ transition rates between these states, and the time and direction of the next transition.
By contrast, policies like MaxWeight may consider an exponential number of schedules at every scheduling decision.
The modulating process for an MSR policy can also be chosen offline during a one-time optimization step that considers the arrival rates, service rates, and resource demands of each job type.

While Theorem \ref{thm:throughput-optimal} tells us that some stable pMSR policy exists whenever some stable policy exists, it does not tell us how to construct a stable pMSR policy. We provide steps to construct a stable pMSR policy, $p$, in Corollary \ref{cor:pmsr-construction}.

\constructionpmsr

\begin{proof}
The capacity constraint ${\bf W}^p{\bf D}\le {\bf P}$ suggests all candidate schedules are in the schedulable set $\mathcal{S}$. Lemma \ref{thm:existence} implies that if $\bm{\pi}^p{\bf W}^p > \bm{\lambda} \oslash \bm{\mu}$, then the system is stable. By Theorem \ref{thm:throughput-optimal}, the optimization problem is feasible.
\end{proof}
While we provide constraints for the optimization problem in Corollary \ref{cor:pmsr-construction}, we have no objective function here because the results of Lemma \ref{thm:existence} did not tell us how to compare two stable pMSR policies. In Section \ref{sec:rt}, we develop analytical tools to optimize over the space of stable MSR policies.

\bgraph{Non-preemptible Jobs.}
We refer to an MSR policy for scheduling non-preemptible jobs as an \emph{nMSR policy}.
Unlike pMSRs, an nMSR policy's modulating process can no longer change states arbitrarily.
Specifically,
we must ensure that no preemptions are required for any state transitions.

Solely increasing the number of type-$i$ jobs scheduled does not require preemptions, because this action simply reserves or reclaims some free resources.
Decreasing the number of type-$i$ jobs scheduled, however, generally requires waiting for some running type-$i$ jobs to complete.
We therefore limit an nMSR policy to decrease the number of scheduled jobs of any type by at most one each time the modulating process changes states.
Because moving to a schedule with one fewer type-$i$ job in service may require completing a type-$i$ job, an nMSR policy $\nmsrPolicy$'s modulating process can transition with rate at most $\mu_i u^{\nmsrPolicy}_i(t)$ in this case.
While these added restrictions clearly break our proof of Theorem \ref{thm:throughput-optimal},
we will argue that the class of nMSR policies remains throughput-optimal.

The key to our argument is to design modulating processes with two types of states: \emph{working states} and \emph{switching states}.
Working states correspond to schedules that actually help stabilize the system.
Switching states exist solely to facilitate switching between working states without preemptions.
It is straightforward to see that, for any pair of working states, there exists a finite sequence of switching states that permit transitions from one working state to the other without preemptions.
For example, a sequence of switching states could first reduce the number of type-1 jobs schedules one at a time, then reduce the number of type-2 jobs, etc.
We define a \emph{switching route} to be a particular sequence of switching states connecting two working states.

To prove the throughput optimality of the class of nMSR policies, we argue that any pMSR policy, $\msrPolicy$, that stabilizes the system can be augmented with some switching routes to create an nMSR policy that stabilizes the system.
We prove this claim in Theorem \ref{thm:nmsr-existence}.
\begin{restatable}{theorem}{nmsrexistence}
    \label{thm:nmsr-existence}
    When scheduling non-preemptible jobs, the class of nMSR policies is throughput-optimal. That is, if there exists a policy that stabilizes a multiresource job system with $K$ job types and arrival rates $\bm{\lambda}$, then there exists an nMSR policy, $\nmsrPolicy$, with $N^{\nmsrPolicy}\leq K$ working states that stabilizes the system. Specifically, an nMSR policy can stabilize a system with non-preemptible jobs when $\bm{\lambda}\in\mathcal{C}$.
\end{restatable}
\begin{proof}
Our proof first invokes Theorem \ref{thm:throughput-optimal} to claim the existence of a $pMSR$ policy, $\msrPolicy$, with $N^{\msrPolicy}\leq K$ that stabilizes the system.
When scheduling non-preemptive jobs, we would like our stationary distribution to be close to that of $\msrPolicy$.
To accomplish this, we let the working states of $\nmsrPolicy$ equal the states of $m^{\msrPolicy}$ and assume arbitrary switching routes.
We then fix the probability that $\nmsrPolicy$ is in each working state \emph{given} that the modulating process in some working state to equal the stationary distribution of $m^{\msrPolicy}$.
By scaling all the transition rates out of the working states proportionally to some switching rate, $\alpha$, we control the amount of time $\nmsrPolicy$ spends in the working states without changing our chosen conditional probabilities.
We then use a Renewal-Reward argument to show that setting $\alpha$ sufficiently low makes the fraction of time spent in any of the working states arbitrarily close to 1.
This implies $\E[\schedule^{\nmsrPolicy}] \approx \E[\schedule^{\msrPolicy}]$.
The full proof of this claim is given in Appendix \ref{app:nmsr}.
\end{proof}

Our proof of Theorem \ref{thm:nmsr-existence} demonstrates the value of preemption.
In order to stabilize the system, our nMSR policies must move through switching states whose corresponding schedules may be highly inefficient.
To avoid spending too much time in these inefficient switching states, an nMSR policy's modulating process must change between working states on a much slower timescale than a pMSR policy, which can directly move between efficient working states.
We will measure the impact of this slower modulation both in theory and in simulation in Sections \ref{sec:rt} and \ref{sec:preemptions}, respectively.

\bgraph{Preemption with Setup Times.}
Our model of MSR policies is even general enough to capture the case where jobs can be preempted, but where preemptions incur a \emph{setup time}.
Here, a setup time represents a period of time during preemptions where server resources are occupied but not actively used by any job.
Setup times generally occur in real systems when the job that is being preempted must be checkpointed and its resources must be reclaimed before the preempting job can start running.
To model this process, we associate every preemption with a i.i.d. drawn random setup time of length $\Gamma \sim \mbox{exp}(\gamma)$.
We refer to $\gamma$ as the \emph{setup rate} of the system.
When a job is preempted, it holds its demanded resources for $\Gamma$ seconds before the preempting job is put in service.
Neither the preempted job nor the preempting job can complete during this setup time.

We call an MSR policy for scheduling jobs with setup times an \emph{sMSR policy}.
The class of sMSR policies shares the challenges with nMSR policies --- the modulating process must be constrained to model preemption overheads.
Hence, we can again divide an sMSR policy's modulating process into working and switching states.
In sMSRs, the transition rates between switching states will depend on the rate $\gamma$ rather than service rates.
Furthermore, the schedules corresponding to each switching state must be altered for jobs that are experiencing setup times but not actually running.
While we will delve more deeply into the performance of sMSR policies in Section \ref{sec:preemptions}, we will begin here by proving the throughput optimality of sMSR policies in Theorem \ref{thm:smsr-existence}.
\begin{theorem}
    \label{thm:smsr-existence}
    When scheduling jobs with setup costs, the class of sMSR policies is throughput-optimal. That is, if there exists a scheduling policy that can stabilize a multiresource job system with $K$ job types and arrival rates $\bm{\lambda}$, then there exists an sMSR policy, $\smsrPolicy$, with $N^{\smsrPolicy}\le K$ working states that stabilizes the system. Specifically, an sMSR policy can stabilize a system with preemptible jobs when $\bm{\lambda}\in\mathcal{C}$.
\end{theorem}
\begin{proof}
    This proof is nearly identical to the proof of Theorem \ref{thm:nmsr-existence}.  While the exact structures of the switching routes between working states differ between the nMSR and sMSR cases, in both cases we choose working states and transitions in order to mirror the behavior of some pMSR policy, $\msrPolicy$, that stabilizes the system.  We use the same argument with the switching rate, $\alpha$, to show that we can construct an sMSR policy based on $\msrPolicy$ that also stabilizes the system.
\end{proof}

Taken together, the throughput optimality theorems of this section show that there is a large design space of MSR policies that can stabilize a multiresource job system.
However, our current tools for analyzing these MSR policies can only tell us whether the system is stable.
To choose a particular MSR policy that achieves low mean response time, we need to develop a stronger analysis of the performance of an MSR policy.
We develop this analytic framework in Section \ref{sec:rt}.

\section{Response Time Bounds for MSR Algorithms}
\label{sec:rt}

Although we have shown that MSR policies are useful for stabilizing the system, it is still unclear how to select an MSR policy that achieves low mean response time.
Hence, in this section, we prove bounds on the mean response of an arbitrary MSR policy $p$.
We will also show how these bounds can be used to provide a closed-form prediction of mean response time under an MSR policy that is highly accurate in practice.
Our prediction formula can then be used to guide the selection of a performant MSR policy.
Specifically, we will use our response time analysis to select a stable MSR policy with low mean response time under a variety of job preemption models in Section \ref{sec:preemptions}.

Our analysis begins by decoupling the original system with $K$ job types into $K$ separate systems.
We will analyze the behavior of these decoupled systems one at a time by comparing the multiresource job system to a simpler single-job-at-a-time system with modulated service rates.
This comparison allows us to prove bounds on the mean queue length for each type of jobs in Theorem \ref{thm:msr-bound}.
We then use our bounds to derive an accurate mean queue length \emph{approximation} in Section \ref{sec:predict}.


\subsection{Decoupling the Multiresource Job System}
We begin by noting that there are two equivalent perspectives from which one can view the multiresource job system.
Perspective 1, as described in Section \ref{sec:model}, is that there is a central queue holding multiresource jobs of all types.
Perspective 2 is that we can imagine $K$ separate queues, one for each type of jobs.
In general, the service rate of the type-$i$ queue can depend on the numbers of jobs in all other queues, so the queues in perspective 2 do not evolve independently, limiting the usefulness of this perspective.

Under an MSR policy, however, Definition \ref{def:MSR} tells us that the number of type-$i$ jobs scheduled by an MSR policy $\msrPolicy$ depends only on the state of the modulating process.
As a result, when considering MSR policies, it is simpler to analyze the expected length of each queue in perspective 2 than it is to analyze the central queue in perspective 1.
Specifically, consider the type-$i$ queue from perspective 2  under an MSR policy $p$.
This queue has Poisson arrivals with rate $\lambda_i$ and jobs whose service requirements are distributed as $\mbox{exp}(\mu_i)$.
The number of jobs allowed in service at time $t$ evolves according to the CTMC $\{u_i^p(t)\}$, which is assumed to be independent of the overall state of the system.
Hence, we call this the \emph{decoupled system} for type-$i$ jobs.
In this section, we will bound the mean queue length in each decoupled system, $\E[Q_i]$.
See Figure \ref{fig:queues} for an example comparing perspective 1 and perspective 2 under an MSR policy.
\begin{figure}[ht]
    \begin{subfigure}[b]{.48\linewidth}
        \centering
        \includegraphics[width=\linewidth,clip,trim={0 7cm 14cm 0}]{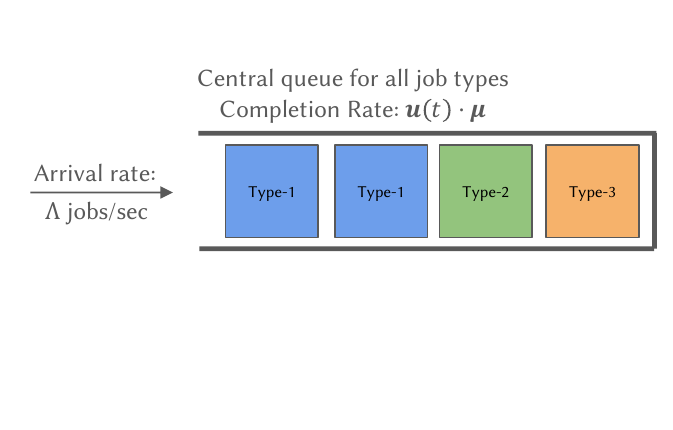}
        \captionsetup{width=\linewidth}
        \caption{Perspective 1: Central queue}
    \end{subfigure}
    \begin{subfigure}[b]{.48\linewidth}
        \centering
        \includegraphics[width=\linewidth,page=2, clip, trim={0 7cm 14.5cm 0}]{images/queues.pdf}
        \captionsetup{width=\linewidth}
        \caption{Perspective 2: Decoupled systems}
    \end{subfigure}
    \caption{Two views of the multiresource job system under an MSR policy. While arriving jobs are stored in a central queue, MSR policies allow us to analyze jobs of each type separately.}
    \label{fig:queues}
\end{figure}

To bound $\E[Q_i]$, we will compare the decoupled system to a similar, simpler system, the \emph{$M/M/1$ Markovian Service Rate system} (MSR-1) system.
A \emph{corresponding MSR-1 system} for type-$i$ jobs is a single server system that serves jobs one at a time and where the service rate of the server varies according to a finite-state CTMC.
Hence, we can construct a corresponding MSR-1 system that looks highly similar to our decoupled system, except that the decoupled system runs multiple type-$i$ jobs in parallel.
We can then analyze the MSR-1 system, and use our analysis to derive bounds on the decoupled system for each job type.
Formally, we define the MSR-1 system as follows.

\begin{restatable}{definition}{decoupledsystem}
   \label{def:msr-1}
   A corresponding MSR-1 system for type-$i$ jobs is an $M/M/1$ queueing system with Markov-modulated service rates.
   Let $\lambda_i$ be the arrival rate of the system.
   The corresponding MSR-1 system shares the same modulating process, $\{m^{\msrPolicy}(t), t\ge 0\}$ as the original MSR system.
   We denote this policy derived from $\msrPolicy$ for the MSR-1 system as $\msronePolicy$.
   Hence, at time $t$, the corresponding MSR-1 system completes jobs at a rate of $u_i^{\msronePolicy}(t)=\mu_i u_i^{\msrPolicy}(t)$.
   Following the definition, $\msrstate{i}{s}=\mu_i\modstate{i}{s}$.
\end{restatable}

Let $\{Q_i^{\msronePolicy}(t), t\ge 0\}$ denote the queue length process of the corresponding MSR-1 system under $\msronePolicy$.
Our goal is to analyze the mean queue length under $\msronePolicy$ and then use this analysis to provide queue length bounds for the original system under MSR policy $\msrPolicy$.

The intuition behind our argument is that, when there are sufficiently many jobs in the system, the MSR and MSR-1 systems complete type-$i$ jobs at the same rate.
When there are not many jobs in the MSR system, the MSR system may complete jobs at a slower rate.
However, we know the queue length of the MSR system is small in this case.
We formalize this argument in Theorem \ref{theorem:msr-difference}.

\msrdifference
\begin{proof}

We prove this claim via a coupling argument in Appendix \ref{app:kextra}.
\end{proof}

Theorem \ref{theorem:msr-difference} tells us that the MSR and MSR-1 systems perform similarly (within an additive constant) at all system loads.
It is worth noting, however, that these bounds become asymptotically tight as the arrival rates in both systems increase and the mean queue length of the MSR-1 system goes to infinity.
This makes intuitive sense --- our proof argues that the only differences between the MSR and MSR-1 systems happen when there are few jobs in one of the two systems, which becomes increasingly rare as the arrival rates increase.
We now proceed to leverage our bounds by analyzing the MSR-1 system in order to bound mean queue length in the MSR system.

\subsection{Analysis of the MSR-1 System}
We now analyze an MSR-1 system running under policy $\msronePolicy$.
Recently, \cite{grosof2023reset} bounded the mean queue length of an MSR-1 system by using the concept of \emph{relative completions}.
The concept was an independent rediscovery of a technique which had previously been introduced by \cite{falin1999heavy} for the Markovian arrival settings, and extended to the Markovian service setting by \cite{dimitrov2011single}, proving equivalent results for mean queue length.
A recent tech report, \cite{grosof2024analysis}, provides a helpful expanded presentation of the technique of \cite{grosof2023reset},
and a more thorough review of the state of the literature. 
For any state $s$ in the modulating process of $\msronePolicy$, the relative completions of the state is defined as the long-run difference in completions between the MSR-1 system that starts from state $s$ and the MSR-1 system that starts from a random state chosen according to the stationary distribution of $\{m^p(t)\}$.
Formally, let $C_i^{\msronePolicy}(t)$ denote the cumulative type-$i$ completions under $\msronePolicy$ by time $t$.
We define $\Delta(s)$ to be the relative completions of state $s$ in the modulating process where
$$\Delta(s)=\lim_{t\rightarrow \infty} C_i^{\msronePolicy}(t\mid m^{\msrPolicy}(0)= s)-\E[u_i^{\msronePolicy}] t.$$
We also define $\wc_i^{\msronePolicy}$ to be the completion-rate-weighted stationary random variable for the state in $\{\wc_i^{\msronePolicy}(t)\}$, which is defined such that $$P\{\wc_i^{\msronePolicy}=s\}=P\{m^{\msrPolicy}=s\}\frac{\msrstate{i}{s}}{\E[u_i^{\msronePolicy}]}.$$
A key result from \cite{grosof2023reset,grosof2024analysis} can then be stated as follows.

\begin{remark}
    \label{remark:marc}
    In \cite{grosof2023reset,grosof2024analysis}, it is shown that
    $$\E[Q_i^{\msronePolicy}]=\frac{\rho_i+\E[\Delta(\wc_i^{\msronePolicy})]}{1-\rho_i}-\E_U[\Delta(\wc_i^{\msronePolicy})],$$
    where $\rho_i=\frac{\lambda_i}{\E[u_i^{\msronePolicy}]}$ and $\E_U[\Delta(\wc_i^{\msronePolicy})]$ denotes the expectation of $\Delta(\wc_i^{\msronePolicy})$ taken over moments when the system experiences unused service.
\end{remark}

\subsection{Queue Length Bounds for the MSR System}
\label{sec:bounds}
Combining Remark \ref{remark:marc} with Theorem \ref{theorem:msr-difference} yields the following bound for the MSR system.
\begin{restatable}{theorem}{msrbound}
    \label{thm:msr-bound}
    For any multiresource job system under an MSR policy, $\msrPolicy$, the mean queue length of type-$i$ jobs is bounded in terms of the corresponding MSR-1 system under policy $\msronePolicy$ as
    $$\E[Q^{\msrPolicy}_i]=\frac{\rho_i+\E[\Delta(\wc_i^{\msronePolicy})]}{1-\rho_i}+B,$$
    where $$-\max_s\Delta(s)\le B\le -\min_s\Delta(s)+\beta_i^{\msrPolicy}.$$
    and
    $\beta_i^{\msrPolicy}$ is the maximum number of type-$i$ jobs served in parallel under $\msrPolicy$. 
\end{restatable}
\begin{proof}
    First, we invoke Theorem \ref{theorem:msr-difference} to relate the mean type-$i$ queue length in the MSR system to a corresponding MSR-1 system, giving
    $$\E[Q_i^{\msronePolicy}]\le \E[Q^{\msrPolicy}_i]\le \E[Q_i^{\msronePolicy}]+\beta_i^{\msrPolicy}.$$
    We then apply Remark \ref{remark:marc} to this bound to obtain
   \begin{equation*}
      \frac{\rho_i+\E[\Delta(\wc_i^{\msronePolicy})]}{1-\rho_i} -  \E_U[\Delta(\wc_i^{\msronePolicy})]  \leq \E[Q_i^{\msrPolicy}]\leq \frac{\rho_i+\E[\Delta(\wc_i^{\msronePolicy})]}{1-\rho_i} -  \E_U[\Delta(\wc_i^{\msronePolicy})] + \beta_i^{\msrPolicy}
    \end{equation*}
    Finally, we note that 
    $$\min_{s} \Delta(s) \leq \E_U[\Delta(\wc_i^{\msronePolicy})] \leq \max_{s} \Delta(s),$$
    which gives the desired bounds.
\end{proof}

To compute the bounds from Theorem \ref{thm:msr-bound} for a given system, one must compute both the distribution of $\wc_i^{\msronePolicy}$ and the extrema of $\Delta$.
Although these quantities do not have a general closed-form, computing the quantities for a given system is straightforward.

First, computing the distribution of $\wc_i^{\msronePolicy}$ amounts to computing the stationary distribution of the modulating process $m^\msrPolicy(t)$ and weighting each term by its completion rate of type-$i$ jobs.
From Theorem \ref{thm:throughput-optimal}, we know that a modulating process needs only $K$ states in order to stabilize the system, so this stationary distribution can generally be computed quickly using standard techniques.

Next, computing the minimum and maximum of the $\Delta$ function requires computing $\Delta(s)$ for each state $s$.
Let $\bm{\Delta}=(\Delta(1), \Delta(2), \cdots, \Delta(N^{\msrPolicy}))^T$.
We compute $\bm{\Delta}$ as follows.

\begin{lemma}
    \label{lemma:solve-delta}
    In an MSR-1 system under policy $\msronePolicy$ whose modulating process has the infinitesimal generator $\bf{G}^{\msronePolicy}$, $\bm{\Delta}$ can be computed as the solution to
    $$
    \begin{cases}
        {\bf G^{\msrPolicy}}\cdot \bm{\Delta}=\E[u_i^{\msronePolicy}]\bm{1}_K - (\msrstate{i}{1}, \msrstate{i}{2}, \cdots, \msrstate{i}{N^{\msronePolicy}})\\
        \E[\Delta(m^{\msrPolicy})] = 0
    \end{cases}.$$
\end{lemma}
\begin{proof}
    We begin with relating the value of $\Delta(s)$ to the $\Delta$-value of its neighboring states: \begin{equation}
        \Delta(s)=\frac{\msrstate{i}{s}-\E[u_i^{\msronePolicy}]}{\sum_{s' \neq s}{\bf G^{\msrPolicy}}_{s,s'}}+\frac{\sum_{s'\neq s} {\bf G^{\msrPolicy}}_{s,s'} \Delta(s')}{\sum_{s' \neq s}{\bf G^{\msrPolicy}}_{s,s'}}. \label{eq:delta-transition}
    \end{equation}
    Equation \ref{eq:delta-transition} is composed of two terms: the relative completions that accrue before the first transition out of state $s$, and the relative completions that accrue after the process leaves $s$.

    Multiplying $\sum_{s' \neq s}{\bf G^{\msrPolicy}}_{s,s'}$ on both sides, we get:
    \begin{align}
        \Delta(s)\sum_{s' \neq s}{\bf G^{\msrPolicy}}_{s,s'}&=\msrstate{i}{s}-\E[u_i^{\msronePolicy}] + \sum_{s'\neq s} {\bf G^{\msrPolicy}}_{s,s'} \Delta(s') \nonumber\\
        \sum_{s'\neq s} {\bf G^{\msrPolicy}}_{s,s'} \Delta(s')-\Delta(s)\sum_{s' \neq s}{\bf G^{\msrPolicy}}_{s,s'}&=\E[u_i^{\msronePolicy}]-\msrstate{i}{s}\nonumber\\
        \sum_{s'}{\bf G^{\msrPolicy}}_{s,s'} \Delta(s')&=\E[u_i^{\msronePolicy}]-\msrstate{i}{s} \label{step:single}\\
        {\bf G^{\msrPolicy}}\cdot \bm{\Delta}&=\E[u_i^{\msronePolicy}]\bm{1}_{N^{\msrPolicy}}-\mu_i(\msrstate{i}{1}, \msrstate{i}{2}, \cdots, \msrstate{i}{N^{\msrPolicy}}) \label{step:vectorize}
    \end{align}
    Step \eqref{step:vectorize} holds because \eqref{step:single} holds for all $s\in[1, N^{\msrPolicy}]$. See \cite{grosof2023reset} for proof that $\E[\Delta(m^p)]=0$.
\end{proof}
Lemma \ref{lemma:solve-delta} tells us that computing all $\Delta$ terms for an MSR-1 policy $\msronePolicy$ amounts to inverting the policy's generator matrix ${\bf G}^{\msrPolicy}$.
Again, note that this will generally be a computationally inexpensive operation due to the limited number of candidate schedules needed by our MSR policies.

\subsection{A Queue Length Approximation for the MSR System}
\label{sec:predict}
Intuitively, the upper bound from Theorem \ref{thm:msr-bound} will be looser at lower loads, since Theorem \ref{theorem:msr-difference} is most accurate at high loads.
Specifically, to obtain a bound in Theorem \ref{thm:msr-bound}, we assumed that the MSR system rarely experiences unused service.
To better model the effects of unused service on the MSR system, we propose the following queue length approximation.

$$\E[Q^{\msrPolicy}_i]\approx P_Q^{M/M/k_i^*}\frac{\rho_i+\rho_i\E[\Delta(\wc_i^{\msronePolicy})]}{1-\rho_i}+\rho_i\cdot k_i^*.$$

Here, $k_i^*=\E[u_i^{\msrPolicy}]$ and $P_Q^{M/M/k_i^*}$ denotes the queueing probability in an $M/M/k_i^*$ system.
We derive this approximation by estimating the relationship between the MSR and MSR-1 systems and applying this estimate instead of Theorem \ref{theorem:msr-difference} in the proof of Theorem~\ref{thm:msr-bound}.
The details of this argument are provided in Appendix \ref{sec:approx}.

\section{Choosing a Stable MSR Policy with Low Response Time}
\label{sec:preemptions}
Section \ref{sec:rt} established results on how to stabilize the system using an MSR policy.
However, these results do not describe how to choose between the many possible MSR policies that might stabilize the system.
In this section, we use our response time analysis from Section \ref{sec:rt} to optimize the design of an MSR policy.
We will consider how to select a good pMSR policy (Section \ref{sec:pmsr}), how to select a good nMSR policy (Section \ref{sec:nmsr}), and how to select a good sMSR policy (Section \ref{sec:smsr}).

To illustrate the effectiveness of our optimization techniques, we will compare the MSR policies we design to the following three state-of-the-art scheduling policies:

\bgraph{MaxWeight} is a preemptive policy proposed in \cite{maguluri2013scheduling} that chooses a schedule by solving a bin-packing problem on every arrival or departure. MaxWeight tries to maximize the weighted number of jobs in service, where the weight of type-$i$ jobs equals to the queue length of type-$i$ jobs.

\bgraph{Randomized-Timers} is a non-preemptive policy proposed in \cite{Ghaderi2016RandomizedCloud} that uses moments of when jobs complete to change schedules.  The policy uses a complex randomized algorithm to discover good solutions to the weighted bin-packing problem used by MaxWeight while the policy runs.

\bgraph{First-Fit} is a simple, non-preemptive policy similar to FCFS.  Sometimes known as BackFilling \cite{Grosof2022WCFS:Systems}, First-Fit examines the queued jobs in FCFS order upon each job arrival or completion.  If a job is found that can be added to the set of running jobs, First-Fit serves this job.  Running jobs are never preempted.  First-Fit is generally thought to be a better version of the FCFS policy from Section \ref{sec:intro}.

We will compare our MSR policies using an example of a server with multiresource jobs.
Specifically, we consider a server with three resources --- 20 cores, 15 GB of DRAM and 50 Gbps of network bandwidth.
We consider a workload with three job types, whose demand vectors are given as 
$${\bf D}_1 = (3,7,1) \quad {\bf D}_2 = (4,1,1) \quad {\bf D}_3=(10,1,5).$$
For this example, we vary $\rho$ by linearly scaling the arrival rate vector, $\bm{\lambda}.$
See Appendix \ref{sec:example} for further details.
While this example is illustrative, we evaluate a more realistic workload in Section \ref{sec:eval-borg}.

\subsection{Choosing a pMSR Policy (Preemptive MSR)}
\label{sec:pmsr}
Corollary \ref{cor:pmsr-construction} provides an Mixed Integer Quadratic Constrained Program (MIQCP) whose solutions correspond to the set of stable pMSR policies.
Our goal is to augment this MIQCP with an objective function to help optimize our choice of pMSR policy.
One approach would be to use either the bounds of Theorem \ref{thm:msr-bound} or our queue length approximation functions directly as the objective functions.
Unfortunately, these functions lack a general closed form, and would make the optimization problem intractable.
We instead propose a simpler objective function to heuristically minimize our queue length upper bounds.

We select a pMSR policy using the following two-step approach.
First, we solve the following MIQCP for the candidate schedules ${\bf W}^p$ and the limiting distribution ${\bm \pi}^p$:
\begin{equation}
    \begin{cases}
    \min & \max_i \rho_i\quad \textnormal{(where $\rho_i=\lambda_i/({\bm\pi}^p{\bf W}^p)_i$)}\\
    \textrm{s.t.} & {\bf W}^p{\bf D} \le {\bf P} \quad \textnormal{(Capacity constraint)}\\
    &\pi^p_i \ge 0, \forall i \quad\textnormal{(Positive fractions)}\\
    &\pi^p_1 + \pi^p_2 + \cdots + \pi^p_{N^p} \le 1 \quad\textnormal{(Fractions sum to 1)}\\
    &\bm{\pi}^p{\bf W}^p > \bm{\lambda} \oslash \bm{\mu} \quad \textnormal{(Stability condition from Lemma \ref{thm:existence})}\\
    &{\bf W}^p \textnormal{ integers}
\end{cases}.\label{eq:optimization}
\end{equation}
Appendix \ref{sec:init-example} gives a numerical example of equation \eqref{eq:optimization}.

Second, we compute a modulating process with a simple loop structure that matches the parameters computed in $\eqref{eq:optimization}$.
Here, we sort the set of candidate schedules lexicographically and let each state transition only to the next neighboring state in this ordering until all states have been visited.
Given this loop structure and the desired values of $\bm{\pi}^p$, we can compute the necessary transition rates by solving a system of $N^p$ equations with $N^p+1$ unknowns.
Solving this system defines the ratios of the transition rates relative to one another, but still allows all transition rates to be scaled by a factor $\alpha>0$.
We refer to $\alpha$ as the \emph{switching rate} of the system because it controls the overall rate of preemptions in the pMSR system.
The loop-structured modulating process with switching rate $\alpha$ for our example system is shown in the Appendix in Figure \ref{fig:mp:pmsr}.

Our bounds suggest that the mean response time of a pMSR policy is minimized by setting $\alpha$ to be arbitrarily high.
Intuitively, the system in this case begins to behave like a non-modulated system using a constant schedule of $\E[\schedule^{\msrPolicy}]$.
We can actually prove this property using our mean queue length bounds for small pMSR instances with two job types (see Appendix \ref{sec:switching-infinitely-often}).
Figure \ref{fig:switch:pmsr} confirms this result by showing how the mean queue length under a pMSR policy improves as $\alpha$ grows.
Notably, although our theoretical results show that taking $\alpha \rightarrow \infty$ is beneficial in this example, Figure \ref{fig:switch:pmsr} shows that we capture most of the performance benefit as long as $\alpha\geq2$.
Hence, we will assume that $\alpha$ is some constant chosen to be sufficiently large.

The above process for selecting a pMSR policy has both intuitive and theoretical justifications.
Intuitively, the $\rho_i$ terms in our bounds denote the load on the system experienced by type-$i$ jobs.
Our bounds show that $\E[Q_i]$ scales as $\Theta(\frac{1}{1-\rho_i})$.
Hence, it is intuitive to use an objective function in \eqref{eq:optimization} that balances load across job types.
Formally, we can prove that the policy we select compares favorably to MaxWeight in heavy traffic.
Specifically, we show that the pMSR policy we select is constant-competitive with MaxWeight in the limit as $\rho \to 1$.
Our proof relies on the fact that both the pMSR and MaxWeight queue lengths scale as $\Theta(\frac{1}{1-\rho})$ in heavy traffic.
Hence, while the second step of our selection process uses a simple loop structure for the modulating process, further optimization of the modulating process only gives constant factor improvement in mean queue length relative to MaxWeight in heavy traffic.

We formally define our notion of heavy-traffic scaling in Definition \ref{def:ht} in Appendix \ref{sec:ht}.  Using this definition, we prove Theorem \ref{thm:ht} as follows.
\begin{restatable}{theorem}{competitiveratio}
\label{thm:ht}
Under the heavy-traffic scaling regime in Definition \ref{def:ht}, the solution to \eqref{eq:optimization} is the same at all loads.
The policy, $p^*$, that we select using \eqref{eq:optimization} is constant-competitive with MaxWeight in heavy traffic.
That is, there exists a constant $C>0$ such that
$$\lim_{\rho \to 1} \frac{\E[Q^p]}{\E[Q^{\text{MaxWeight}}]} \leq C$$
\end{restatable}
\begin{proof}
See Appendix \ref{sec:ht} for a full proof.
\end{proof}

\begin{figure}
    \begin{subfigure}[b]{.48\linewidth}
        \centering
        \includegraphics[width=\linewidth]{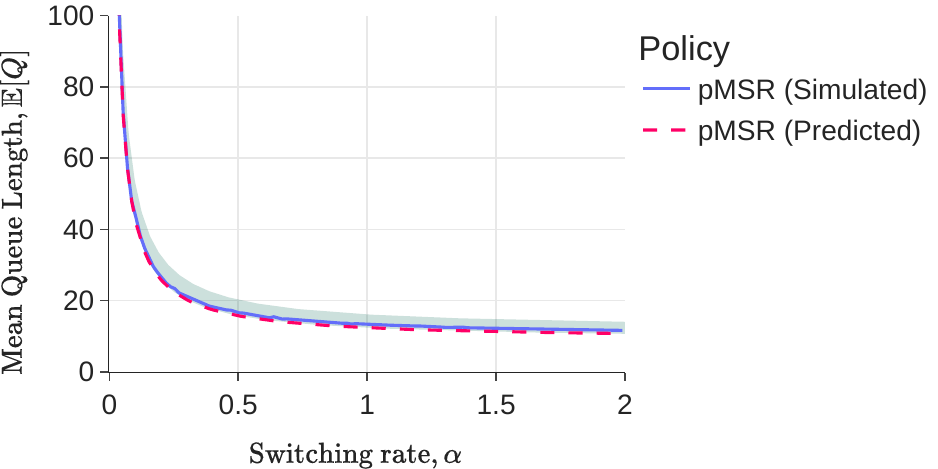}
        \captionsetup{width=\linewidth}
        \caption{pMSR}
        \label{fig:switch:pmsr}
    \end{subfigure}
    \begin{subfigure}[b]{.48\linewidth}
        \centering
        \includegraphics[width=\linewidth]{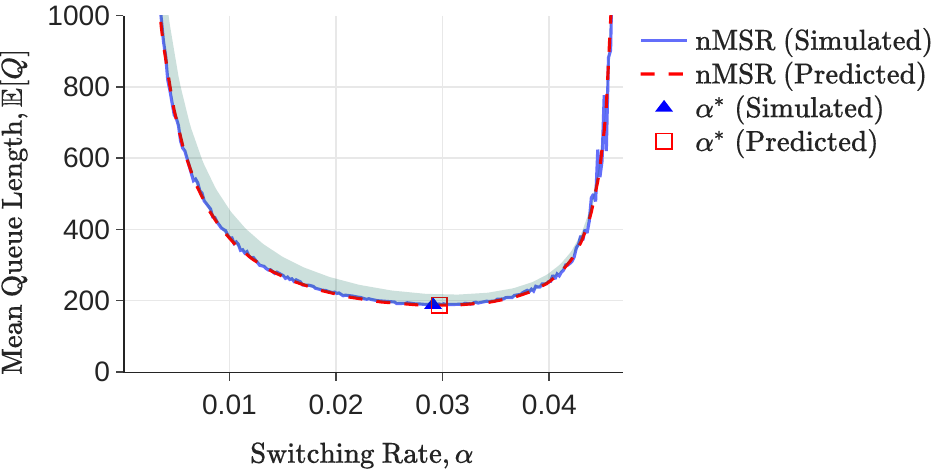}
        \captionsetup{width=\linewidth}
        \caption{nMSR}
        \label{fig:switch:nmsr}
    \end{subfigure}
    \caption{The effect of switching rate on MSR policies under various preemption behaviors in the example from Section \ref{sec:preemptions}. The shaded regions depict our queue length bounds.  Our pMSR policies benefit from a high switching rate when system load $\rho=0.9$.  However, nMSR policies suffer when $\alpha$ is too high or too low.  Our queue length prediction is accurate and can be used to select $\alpha^*$ in both cases.}
\end{figure}

\subsection{Choosing an nMSR Policy (Nonpreemptive MSR)}
\label{sec:nmsr}
Given the choice of pMSR policy in Section \ref{sec:pmsr}, the proof of Theorem \ref{thm:nmsr-existence} yields a related stable nMSR policy.
Specifically, our proof suggests setting the working states to be the states of $m^{\msrPolicy}(t)$ and then using arbitrary switching routes between these states.
To stabilize the system, we can scale the transition rates out of the working states by a very small switching rate $\alpha > 0$.
This method for choosing an nMSR policy has two main issues.
First, switching routes are chosen arbitrarily instead of being chosen to optimize mean response time.
Second, while $\alpha$ must be small to guarantee stability, the optimal switching rate $\alpha^*$ is not obvious here.

The space of possible switching routes between working states is massive.
Because our queue length bounds must be recomputed for each choice of switching routes, it is hard to optimize this choice using our results.
However, we sampled many possible switching routes for the example system in this section and found that the choice of switching routes had a minimal impact on mean queue length.
Hence, we defer a full analysis of optimizing switching routes to future work.
An example of an nMSR policy with simple switching routes is shown in Figure \ref{fig:nmsr-states} of Appendix \ref{sec:example}.

Our queue length bounds \emph{are} useful for optimizing the choice of $\alpha$.
Given the pMSR policy, $\msrPolicy$, selected in Section \ref{sec:pmsr}, let $\nmsrPolicy(\alpha)$ be the nMSR policy based on $\msrPolicy$ whose transitions out of working states are scaled by $\alpha > 0$.
Because it is non-preemptive, $\nmsrPolicy(\alpha)$ can suffer when $\alpha$ is high and the policy spends too much time in switching states.
Conversely, if $\alpha$ is too small, the relative completions terms in our queue length analysis can become large.
We therefore use our bounds to select the optimal switching rate, $\alpha^*$, that balances the cost of spending time in switching states with the cost of switching too slowly between working states.

Figure \ref{fig:switch:nmsr} shows that our mean queue length predictions allow us to accurately predict $\alpha^*$, and that choosing the correct value of $\alpha$ dramatically reduces mean queue length.
Specifically, even though we chose $\alpha$ to minimize our  queue length prediction, the value that minimizes this prediction ($\alpha^*$ {\bf Predicted}) is within $2.1\%$ of the $\alpha^*$ derived in simulation ($\alpha^*$ {\bf Simulated}).
The result of using $\alpha^*$ {\bf Predicted} is a mean queue length that is within $1.5\%$ of optimal.

Furthermore, we can bound the potential error from minimizing our prediction function instead of an exact formula for $\E[Q]$.
The $\E[Q]$ from using $\alpha^*$ {\bf Predicted} is at most the minimum value of our upper bound over all $\alpha$'s.
Similarly, the $\E[Q]$ under $\alpha^*$ {\bf Simulated} is at least the minimum value of our \emph{lower} bound over all $\alpha$'s.
In our example, the difference between the minima of our upper and lower bounds is $16.3\%$.
Hence, while our prediction function is more accurate in practice, we can \emph{guarantee} that using $\alpha^*$ {\bf Predicted} is within $16.3\%$ of the nMSR policy with the optimal $\alpha$.
By contrast, Figure \ref{fig:switch:nmsr} shows that choosing the wrong $\alpha$ can easily increase $\E[Q]$ by an order of magnitude or even result in instability.
We can use our $\E[Q]$ bounds similarly to provide error bounds whenever optimizing parameters of an MSR policy using our prediction function.

\subsection{Choosing an sMSR Policy (MSR with Setup)}
\label{sec:smsr}
To select an sMSR policy, $\smsrPolicy$, we closely follow our selection process for the nMSR policy, $\nmsrPolicy$.
The central difference between these policies is the structure of their switching routes, which differ in two key ways.
First, the rates that the policies move through switching states will differ because $\nmsrPolicy$ switches by waiting for completions with rate $\mu_i$ and $\smsrPolicy$ switches by waiting for setup times with setup rate $\gamma$.
Second, the nMSR policy completes jobs during its switching states, but the time $\smsrPolicy$ spends in its switching states is wasted because it is unable to complete the jobs being preempted.
Nonetheless, again consider a range of sMSR policies, $\smsrPolicy(\alpha)$, whose working state transitions are scaled by $\alpha$, using our analysis to select an optimal switching rate $\alpha^*$.

\begin{wrapfigure}{l}{0.48\textwidth}
    \centering
    \includegraphics[width=\linewidth]{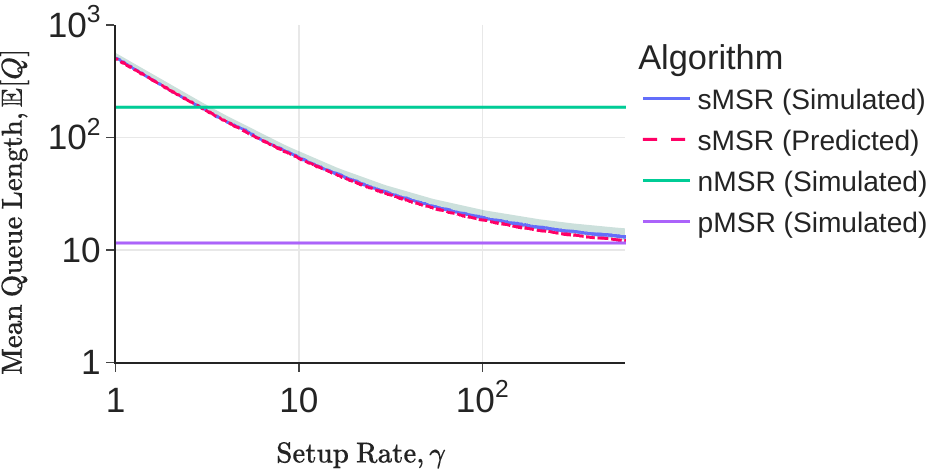}
    \caption{Mean queue length of the sMSR policy $\smsrPolicy(\alpha^*)$ as a function of setup rate, $\gamma$ when $\rho=0.9$. 
    The shaded regions depict our queue length bounds for $\smsrPolicy(\alpha^*)$. 
    We compare the performance of $\smsrPolicy(\alpha^*)$ to the nMSR policy, $\nmsrPolicy(\alpha^*)$, and the pMSR policy, $\msrPolicy$.
    }
    \label{fig:smsr}
\end{wrapfigure}
Because one can always use an nMSR policy and avoid costly preemptions altogether, it is natural to compare our sMSR policy, $\smsrPolicy(\alpha^*)$, to the nMSR policy $\nmsrPolicy(\alpha^*)$.
Specifically, one must understand how fast the setup times must be in order for $\smsrPolicy(\alpha^*)$ to outperform $\nmsrPolicy(\alpha^*)$.

Figure \ref{fig:smsr} shows that our mean queue length bounds accurately predict when it is better to use an sMSR policy than an nMSR policy.
The $\smsrPolicy(\alpha^*)$ policy performs better than $\nmsrPolicy(\alpha^*)$ when $\gamma$ is significantly higher than $\mu_i$ for all $i$, as expected.
As $\gamma$ grows and preemption overhead decreases,  we see the performance of $\smsrPolicy(\alpha^*)$ approach the performance of $\msrPolicy$.
In these experiments, we used our mean queue length bounds to select a different $\alpha^*$ for each policy and for each value of $\gamma$.

\subsection{The Example System}
\label{sec:simple}
Figure \ref{fig:simple} summarizes our results from this section by measuring the performance of our optimized MSR policies and the competitor policies as a function of system load.

Figure \ref{fig:simple:pmsr} shows our pMSR policy, $\msrPolicy$, is close to the performance of MaxWeight.
While MaxWeight beats $\msrPolicy$ at moderate loads, they appear to converge as $\rho\rightarrow 1$.
First-Fit performs well at $\rho<.9$ but becomes unstable around $\rho=.9$.
The non-preemptive Randomized-Timers lags far behind $\msrPolicy$.

\begin{figure}[b]
    \centering
    \begin{subfigure}[b]{.27\linewidth}
        \centering
        \includegraphics[width=\linewidth]{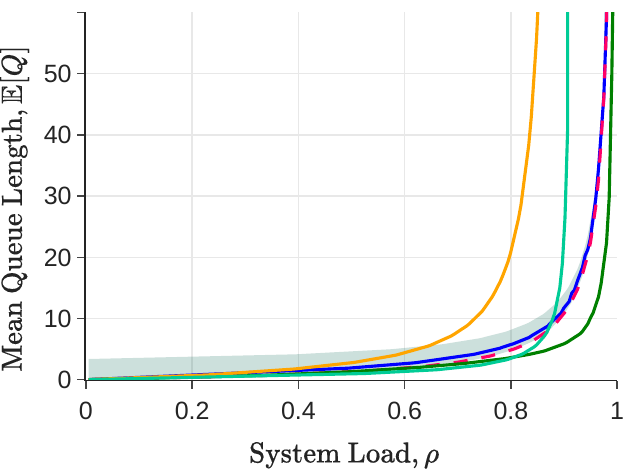}
        \captionsetup{width=.8\linewidth}
        \caption{pMSR}
        \label{fig:simple:pmsr}
    \end{subfigure}
    \begin{subfigure}[b]{.27\linewidth}
        \centering
        \includegraphics[width=\linewidth]{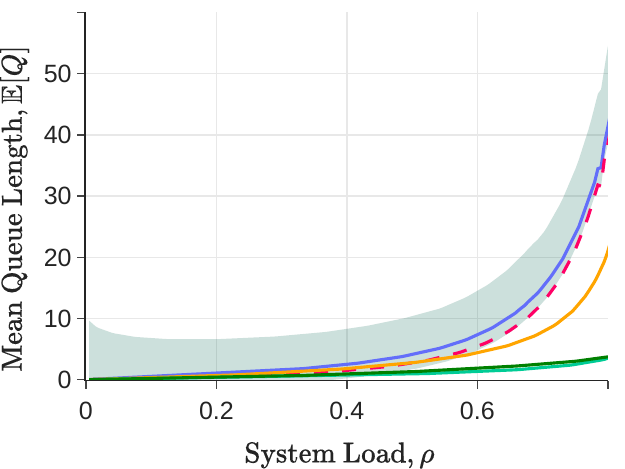}
        \captionsetup{width=1\linewidth}
        \caption{nMSR at low loads}
        \label{fig:simple:nmsr:low} 
    \end{subfigure}
    \begin{subfigure}[b]{.41\linewidth}
        \centering
        \includegraphics[width=\linewidth]{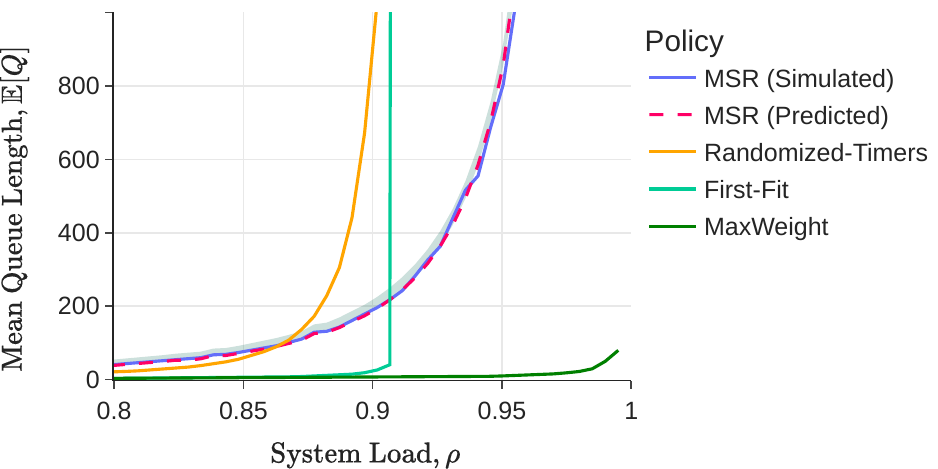}
        \captionsetup{width=.8\linewidth}
        \caption{nMSR at high loads}
        \label{fig:simple:nmsr:high} 
    \end{subfigure}
    \caption{Mean queue length under pMSR and nMSR policies in the example from Section \ref{sec:preemptions}.
    The pMSR policy uses $\alpha=2$. We use $\alpha^*$ for the nMSR policy at each load.
    The shaded regions depict the queue length bounds for our MSR policies. The dashed lines represent our queue length approximations.}
    \vspace{-.15in}
    \label{fig:simple}
\end{figure}

We also evaluate the nMSR policy, $\nmsrPolicy(\alpha^*)$, in Figures \ref{fig:simple:nmsr:low} and \ref{fig:simple:nmsr:high}.
MaxWeight greatly outperforms $\nmsrPolicy(\alpha^*)$ at all loads because it is preemptive.
Furthermore, $\nmsrPolicy(\alpha^*)$ is not the best of the non-preemptive policies at low loads.
However, $\nmsrPolicy(\alpha^*)$ becomes the best non-preemptive policy under high load.
Here, First-Fit becomes unstable, while Randomized-Timers suffers from changing between schedules too slowly.
In fact, Randomized-Timers performed so poorly under high load that many of our simulations did not converge (these points were omitted).


\section{Evaluation}
\label{sec:eval-borg}
Section \ref{sec:preemptions} shows how to optimize the performance of MSR policies under our theoretical model.
This raises the question of how our policies perform under a real-world workload that violates several of our modeling assumptions (see Section \ref{sec:assumptions}).
To address this question, we performed simulations based on publicly available data from the Google Borg system \cite{tirmazi2020borg}\footnote{Code is available online at \href{https://github.com/jcpwfloi/msr-borg}{https://github.com/jcpwfloi/msr-borg}}.

\subsection{Experimental Setup}
Using the data from \cite{tirmazi2020borg}, we constructed a trace of 69 million requests that were submitted to Borg over a 30-day period in 2019.
The trace contains ten job types, and jobs demand both CPU cores and memory.
Job sizes and interarrival times follow highly variable distributions that are far from our theoretical model.
For details on the construction and features of this trace, see Appendix \ref{sec:trace}.

To simulate different arrival rates, we downsampled the trace by various downsampling factors.
Following the notation of Section \ref{sec:model}, we refer to our downsampling factor as \emph{system load}, $\rho$, and normalize this factor so a value of $\rho < 1$ implies that at least one policy we tested was stable.

Our goal is to compare the mean response time, $\E[T]$, under different scheduling policies as a function of $\rho$.
The competitor policies we consider follow the policy definitions of Section \ref{sec:preemptions}.
Specifically, we once again consider the MaxWeight, Randomized-Timers, and First-Fit policies.
For the MSR policy variants, we derive the set of working states by solving the optimization problem in \eqref{eq:optimization}.
We again determine transition rates by assuming a loop-structured modulating process with an overall switching rate of $\alpha$.
For the pMSR policies we consider, we set $\alpha=0.1$.
For the nMSR policies, we set $\alpha=3.6\times10^{-6}$, a value chosen to ensure stability when $\rho=.9$.

\subsection{A Practical Optimization: MSR With BackFilling}
\label{sec:backfill}
All of our MSR policy variants suffer during moments of unused service.
Here, resources earmarked by the modulating process to serve type-$i$ jobs are left idle when there are not enough type-$i$ jobs in the system.
This wasted service has two impacts on our response time bounds.
First, there are the additive terms in our bounds that correspond directly to the unused service for each job type.
Reducing the amount of unused service will reduce these terms.
Second, unused service also imposes an opportunity cost on the system --- the wasted capacity for class $i$ jobs could have been used to serve some other type of jobs, $j$, reducing the $\Theta\left(\frac{1}{1-\rho_j}\right)$ term in our bound.
Hence, particularly as system load increases, reducing unused service even slightly can have an outsize benefit on mean response time by reducing the load experienced by other job types.

We therefore simulate variants of our MSR policies that include a practical optimization for filling unused service capacity in the system.
If an MSR policy tries to serve a class $i$ job but does not find one available, we will try to fill the excess capacity by serving additional jobs from the central queue according to the First-Fit policy.
We refer to these MSR variants as \emph{MSR with BackFilling} policies.
When a pMSR with BackFilling policy changes working states, it can simply preempt any additional jobs that were put into service by BackFilling, thus guaranteeing that all jobs will finish no later under pMSR with BackFilling as compared to the original pMSR policy.
Therefore, the throughput-optimality results and response time upper bounds for pMSR policies extend to the analogous pMSR with BackFilling policies.
When an nMSR or sMSR with BackFilling policy changes working states, however, additional jobs may need to be preempted or completed in order to change states, making it hard to guarantee a performance benefit from BackFilling.
Nonetheless, our simulations show that even with non-preemptible jobs the benefits of BackFilling outweigh the extra time the system spends in switching states.
A more rigorous analysis of MSR policies with BackFilling will be addressed in future work.
\begin{figure}
    \centering
    \includegraphics[width=0.6\linewidth]{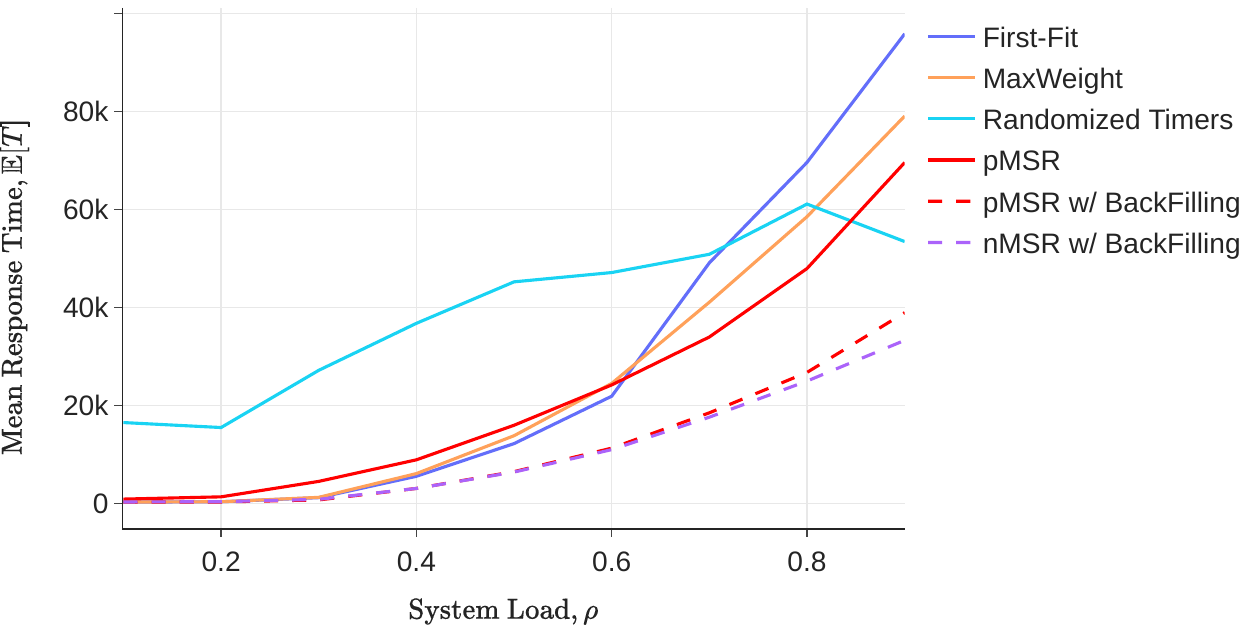}
    \caption{Evaluation of MSR policies using a Google Borg trace. The nMSR policy without BackFilling is not pictured, but performed roughly 10x worse than the policies shown due to excessive unused service.  The pMSR and nMSR policies with BackFilling provide low mean response times at all loads.}
    \label{fig:borg}
    \vspace{-.1in}
\end{figure}
\subsection{Experimental Results}
Our simulation results comparing scheduling policies from the literature to a variety of MSR policies using the Google Borg trace are shown in Figure \ref{fig:borg}.
Notably, the nMSR policy we tested is not pictured because it performs roughly 10x worse than the policies in Figure \ref{fig:borg} due to its high degree of unused service.
While First-Fit and MaxWeight perform well at low loads, they are beaten by the pMSR policy under moderate loads.  The pMSR policy outperforms First-Fit and MaxWeight by 40\% and 15\% respectively when $\rho=.8$.
Notably, the MSR with BackFilling variants of the pMSR and nMSR policies provide excellent mean response times, outperforming both First-Fit and MaxWeight by a factor of 2 when load is high.
While Randomized-Timers lags behind the MSR with BackFilling policies, it begins to be competitive at higher loads.
The non-monotonicity of the Randomized-Timers data is caused by the difficulty of tuning various policy parameters.

It is worth noting that the optimization code used to compute the candidate set for the MSR policies completed in just 15.28 seconds on a modest 28-core server using the Gurobi solver.
Unlike MaxWeight, this optimization code is not on the critical path of the MSR scheduling policies, and thus could be continuously re-run in the background using modest computational resources.

\section{Conclusion}
 This paper devises a class of low-complexity scheduling policies for multiresource jobs called MSR policies.
An MSR policy selects a set of candidate schedules to switch between using a CTMC.
The class of MSR policies is throughput-optimal despite its simplicity.
We show how to select an MSR policy that achieves low mean response time in a variety of practical cases.

While our analysis of MSR policies yields accurate predictions of mean queue length, MSR policies have some limitations.
Specifically, MSR policies are outperformed by MaxWeight and First-Fit at moderate loads.
These competitor policies have complex, infinite state modulating processes that depend on which jobs are in the queue.
Our current analysis does not apply to these complex cases, but a natural direction for future work is to design policies that incorporate some state-dependent behavior while remaining tractable.
Specifically, we are encouraged by our simulations of MSR with BackFilling, and will aim to generalize our analysis to handle these policies.

\begin{acks}
  We thank our shepherd, Dr. Amr Rizk, and the anonymous reviewers for their insightful feedback that helped improve our work.
  This work is supported by a Northwestern IEMS Startup Grant, AFOSR Grant FA9550-24-1-0002, a UNC Chapel Hill Startup Grant, and National Science Foundation grants NSF-CCF-2403195 and NSF-IIS-2322974.
\end{acks}

\clearpage
\bibliographystyle{ACM-Reference-Format}
\bibliography{references,manual-reference,bibshort}

\appendix
\section*{Appendix}

\section{Proof of Lemma \ref{thm:existence}}
\label{sec:stability}
\existence*
\begin{proof}
    These criteria are clearly necessary for stability, since it is easy to see that when the condition is violated, $\E[Q^p_i]$ is infinite for at least one job type, $i$.

    We must then show that our condition is sufficient for stability.
    We first define a test function $V({\bf Q}(t))$ to capture the system state as a number, $$V({\bf Q}(t))=\norm{\bf Q}_2^2.$$
    For convenience of notation, we set $V(t)=V({\bf Q}(t))$.
    Intuitively, this test function is a measure of queue length. The larger the test function is, the longer the overall queue is.
    Then, we invoke Foster-Lyapunov Theorem \cite{Stolyar2004MaxWeightTraffic} and show that $\{{\bf Q}(t), t\ge 0\}$ is positive-recurrent.
    Let $\Delta V({\bf q})$ denote the \emph{drift} at state ${\bf q}$, where drift denotes the instantaneous change rate of the test function when the current state is ${\bf q}$. Formally, $$\Delta V({\bf q})=\lim_{t\rightarrow 0}\frac{1}{t}\E[V(t)-V(0)\mid {\bf Q}(0)={\bf q}].$$
    Foster-Lyapunov Theorem states that when the drift function is negative for all but a finite set of exception states, the Markov chain $\{{\bf Q}(t), t\ge 0\}$ is positive-recurrent.
    \begin{remark}[Foster-Lyapunov Theorem] \label{remark:foster-lyapunov} Let $\mathcal{M}$ be a finite subset of $\mathbb{Z}_+^K$. If there exists constants $\epsilon, b>0$, such that 
    
    $$\forall {\bf q}\in\mathbb{Z}_+^K, \lim_{t\rightarrow0} \Delta V({\bf q})\le -\epsilon + b\mathbf{1}_{{\bf q}\in\mathcal{M}},$$
    then $\{{\bf Q}(t), t\ge 0\}$ is positive recurrent.
    \end{remark}
    We now bound the drift of our test function $V$. For some constant $c_1>0$,
    \begin{align}
        \Delta V({\bf q})&=\lim_{t\rightarrow 0}\frac{1}{t}\E[V(t)-V(0)\mid {\bf Q}(0)={\bf q}]\nonumber\\
        &=\lim_{t\rightarrow 0}\frac{1}{t}\E[\norm{{\bf q}+{\bf A}(t)-\hat{\Completions}(t)}_2^2-\norm{\bf q}_2^2\mid {\bf Q}(0)={\bf q}] \nonumber\\
        &\le \lim_{t\rightarrow 0}\frac{1}{t}\E[\norm{{\bf q}+{\bf A}(t)-\Completions(t)}_2^2-\norm{\bf q}_2^2\mid {\bf Q}(0)={\bf q}] \label{eq:unused}\\
        &=\lim_{t\rightarrow 0}\frac{1}{t}\E[\norm{{\bf A}(t)-\Completions(t)}_2^2+2\langle {\bf q}, \left({\bf A}(t)-\Completions(t)\right)\rangle\mid {\bf Q}(0)={\bf q}] \nonumber\\
        &=\lim_{t\rightarrow 0}\frac{1}{t}\left(\E[\norm{{\bf A}(t)-\Completions(t)}_2^2]+2\langle {\bf q}, \E[{\bf A}(t)-\Completions(t)\mid {\bf Q}(0)={\bf q}]\rangle\right)\nonumber\\
        &\le\lim_{t\rightarrow 0}c_1^2K+\frac{2}{t}\langle {\bf q}, \E[{\bf A}(t)-\Completions(t)\mid {\bf Q}(0)={\bf q}]\rangle.\label{eq:bounded-arrivals}
    \end{align}
    (\ref{eq:unused}) follows from the fact that when $\hat{\Completions}_i(t)<\Completions_i(t)$, i.e., when unused service occurs, ${\bf q}_i+{\bf A}_i(t)-\hat{\Completions}_i(t)=0, \forall i$, therefore setting $\hat{\Completions}(t)$ to $\Completions(t)$ won't decrease the norm.
    (\ref{eq:bounded-arrivals}) follows from the fact that $\E\left[\frac{{\bf A}(t)}{t}\right]=\bm{\lambda}$ and $\E\left[\frac{\Completions(t)}{t}\right]\le \E\left[\frac{{\bf A}(t)}{t}\right]$.
    Next, we simplify \eqref{eq:bounded-arrivals}:
    \begin{align}
        \lim_{t\rightarrow0}\frac{2}{t}\langle {\bf q}, \E[{\bf A}(t)-\Completions(t)\mid {\bf Q}(0)={\bf q}]\rangle&=\lim_{t\rightarrow0}\frac{2}{t}\langle {\bf q}, \E[{\bf A}(t)\mid {\bf Q}(0)={\bf q}]-\E[\Completions(t)\mid {\bf Q}(0)={\bf q}]\rangle\nonumber\\
        &=\lim_{t\rightarrow0}\frac{2}{t}\langle {\bf q}, \E[{\bf A}(t)]-\E[\Completions(t)]\rangle\label{eq:independent}\\
        &=\lim_{t\rightarrow0}\frac{2}{t}\langle {\bf q}, \bm{\lambda} t- \bm{\mu} t \otimes \E[\schedule^p]\rangle\label{eq:avg}\\
        &=2\langle {\bf q}, \bm{\lambda}-\bm{\mu}\E[\schedule^p]\rangle\nonumber
    \end{align}
    (\ref{eq:independent}) follows from the assumption that $\{{\bf A(t)}\}$ and $\{{\bf C(t)}\}$ are independent of queue length. (\ref{eq:avg}) follows from the fact that the expected service rate is $\bm{\mu}\otimes \E[\schedule^p]$.
    Therefore, if $\bm{\lambda} < \bm{\mu}\otimes \E[\schedule^p]$, then there exists an $\epsilon, b>0$ and some $\mathcal{M}$ such that
    $$\forall {\bf q}\in\mathbb{Z}_+^K, \Delta V({\bf q})\le -\epsilon + b\mathbf{1}_{{\bf q}\in\mathcal{M}}.$$
    Hence, by Remark \ref{remark:foster-lyapunov}, $\{{\bf Q}(t), t\ge0\}$ is a positive recurrent CTMC if $\bm{\lambda} < \bm{\mu}\otimes \E[\schedule^p]$.
\end{proof}
\section{Proof of Theorem \ref{theorem:msr-difference}}
\label{app:kextra}
\msrdifference*
\begin{proof}
We prove this claim via a coupling argument.
First, observe that the arrival processes and modulating processes are identical between the MSR and MSR-1 systems.
Hence, to compare the two systems on a given sample path, we will assume that arrival times are the same in both systems and that both systems' modulating processes are in corresponding states at any time $t$.

Next, we couple the departure processes of the two systems as follows.  
At any time, $t$, the MSR system attempts to serve $\schedule^{\msrPolicy}_i(t)$ type-$i$ jobs in parallel.
Similarly, the service rate of type-$i$ jobs in the MSR-1 system is $\mu_i\schedule^{\msrPolicy}_i(t)$.
To couple the systems, we set $\schedule^{\msrPolicy}_i(t)$ exponential timers at any time $t$ corresponding to an event (arrival, departure, or change to the modulating process) in either system.
If one of these timers expires before the next event, it triggers a type-$i$ completion in the MSR-1 system as long as there is at least one job to complete.
In the MSR system, we map the type-$i$ jobs in service at time $t$ to one timer each.
Note that this means some timers may not have an associated job.
If a timer expires that is associated with a job in the MSR system, this associated job completes.

We will show that, for any sample path under our coupling, 
\begin{equation}Q_i^{\msronePolicy}(t)\le Q^{\msrPolicy}_i(t) \le Q_i^{\msronePolicy}(t)+\beta_i^{\msrPolicy} \quad \forall t\geq0 \label{eq:induct}\end{equation}
We prove this claim by induction on $t$.
Without loss of generality, we assume that the claim holds at time 0 where $Q_i^{\msronePolicy}(0)= Q^{\msrPolicy}_i(0)=0$.
Consider any time $t$ where \eqref{eq:induct} is assumed to hold for induction.
We will consider two cases.

First, consider the case where $Q_i^{\msrPolicy}(t) < \beta_i^{\msrPolicy}$.
Let $t'$ be the next time after time $t$ where $Q_i^{\msrPolicy}(t) = \beta_i^{\msrPolicy}$.
In the interval $[t,t']$, any completion in the MSR system corresponds to a completion in the MSR-1 system, unless the MSR-1 system has 0 type-$i$ jobs.
Hence, $Q_i^{\msronePolicy}(t)\le Q^{\msrPolicy}_i(t)$ implies that at each successive completion time $c\in[t,t']$, $Q_i^{\msronePolicy}(c)\le Q^{\msrPolicy}_i(c)$.
Because arrivals affect both systems equally, we have that $Q_i^{\msronePolicy}(\tau)\le Q^{\msrPolicy}_i(\tau)$ for each arrival time $\tau\in [t,t']$.
Finally, because $Q_i^{\msrPolicy}(\tau) \leq \beta_i^{\msrPolicy}$ for $\tau\in [t,t']$, it is trivial to see that $Q_i^{\msrPolicy}(\tau) \leq Q_i^{\msronePolicy}(\tau) + \beta_i^{\msrPolicy}$ for $\tau\in [t,t']$.

Second, consider the case where $Q_i^{\msrPolicy}(t) \geq \beta_i^{\msrPolicy}$.
Let $t'$ be the next time after $t$ where $Q_i^{\msrPolicy}(t) < \beta_i^{\msrPolicy}$.
In the interval $[t,t']$, any completion in the MSR system corresponds to a completion in the MSR-1 system, unless the MSR-1 system has 0 type-$i$ jobs.
This once again implies that the first part of the inequality holds on the interval $[t,t']$.
To prove the second part of the inequality, note that now any completion in the MSR-1 system necessarily corresponds to a completion in the MSR system, since the MSR system has a job associated with every timer for the entire interval.
Hence, $Q_i^{\msrPolicy}(t)\le Q^{\msronePolicy}_i(t) + \beta_i^{\msrPolicy}$ implies that at each successive completion time $c\in[t,t']$, $Q_i^{\msrPolicy}(c)\le Q^{\msronePolicy}_i(c) + \beta_i^{\msrPolicy}$ because every decrease in $Q^{\msronePolicy}_i$ corresponds to a simultaneous decrease in $Q^{\msrPolicy}_i$.
Thus, $Q_i^{\msrPolicy}(\tau) \leq Q_i^{\msronePolicy}(\tau) + \beta_i^{\msrPolicy}$ for $\tau\in [t,t']$.

Hence, given that \eqref{eq:induct} holds at time $t$, it also holds on the entire interval $[t,t']$, where $t'$ is at least one event time later than $t$.
This completes the proof of \eqref{eq:induct} by induction.

Given that $\eqref{eq:induct}$ holds for any sample path under our coupling, we have
$$Q_i^{\msronePolicy}\le_{st} Q^{\msrPolicy}_i \le_{st} Q_i^{\msronePolicy}+\beta_i^{\msrPolicy}$$
and thus
$$\E[Q_i^{\msronePolicy}]\le \E[Q^{\msrPolicy}_i] \le \E[Q_i^{\msronePolicy}]+\beta_i^{\msrPolicy}$$
as desired.
\end{proof}
\section{A Queue Length Approximation for the MSR System}
\label{sec:approx}

The slack in the bounds of Theorem \ref{thm:msr-bound} results from the complexity of comparing an $M/M/k$ system with modulated service rates (the MSR system) to an analogous $M/M/1$ system with modulated service rates (the MSR-1 system).
However, when service rates are not modulated, this comparison follows the well-known form
\begin{equation}\E[N^{M/M/k}] = P^{M/M/k}_Q \cdot \E[N^{M/M/1}] + k\rho^{M/M/1}\label{eq:mmk}\end{equation}
where $P^{M/M/k}_Q$ is the probability an arriving job queues in the $M/M/k$ system, $\rho^{M/M/1}$ is the system load of the $M/M/1$ system, and hence the second term in the sum represents the average number of jobs in service in the $M/M/k$ system.
When service rates are generally distributed, previous works \cite{psounis2005systems} also suggest that this difference is difficult to exactly model. 
To derive Theorem \ref{theorem:msr-difference}, we noticed that the MSR and MSR-1 systems perform similarly \emph{during periods where both systems have sufficiently many jobs in their queues}.
As a result, the upper bound from Theorem \ref{theorem:msr-difference} follows a similar form as \eqref{eq:mmk}, but pessimistically bounds $P_Q$ by 1 and bounds the expected number of jobs in service by $\beta_i^{\msrPolicy}$ to relate the MSR and MSR-1 systems.

Although we cannot provide a tighter \emph{bound} on these two quantities, we can provide \emph{approximations} of these quantities, resulting in a queue length approximation that performs well at all loads.
To approximate $P_Q$, we treat the MSR system as if it always tries to serve $k_i^*=\E[u_i^{\msrPolicy}]$ jobs, setting $P_Q=P_Q^{M/M/k_i^*}$.
We approximate the average number of type-$i$ jobs in service as $\rho_i \cdot k_i^*$.
By applying these approximations instead of applying Theorem \ref{theorem:msr-difference} in the proof of Theorem \ref{thm:msr-bound}, we get
$$\E[Q^{\msrPolicy}_i]\approx P_Q^{M/M/k_i^*}\frac{\rho_i+\E[\Delta(\wc_i^{\msronePolicy})]}{1-\rho_i} - \E_U[\Delta(\wc_i^{\msronePolicy})] +\rho_i\cdot k_i^*.$$

Finally, as suggested in \cite{grosof2023reset}, we approximate $\E_U[\Delta(\wc_i^{\msronePolicy})]$ as $\E[\Delta(\wc_i^{\msronePolicy})]$, yielding the following approximation of mean queue length for an MSR-1 system

$$\E[Q^{\msrPolicy}_i]\approx P_Q^{M/M/k_i^*}\frac{\rho_i+\rho_i\E[\Delta(\wc_i^{\msronePolicy})]}{1-\rho_i}+\rho_i\cdot k_i^*.$$

If our prediction lies below or above our lower or upper bounds, respectively, we report the value of the closest bound.
Our prediction and bounds are compared to simulations in Section \ref{sec:preemptions}.

\section{Switching Rates for Preemptive MSR Policies}
\label{sec:switching-infinitely-often}
\begin{theorem}
    For a two-state pMSR policy $p$, where $N^p=2$, the mean queue length upper bound is minimized if we are switching schedules infinitely often. That is, $\alpha=\infty$ minimizes the mean length upper bound across all $p(\alpha)$'s.
\end{theorem}
\begin{proof}
    Let $\alpha{\bf G}^p=\begin{pmatrix}
        -\alpha x_1 & \alpha x_1\\
        \alpha x_2 & -\alpha x_2
    \end{pmatrix}$ for some $x_1, x_2>0$.
    WLOG, we assume $\msrstate{1}{1}\ge \msrstate{1}{2}$. 
    Let $w_1, w_2$ be two working states of the system.
    By Lemma \ref{lemma:solve-delta}, we get \begin{equation*}
        \begin{cases}
            \begin{pmatrix}
                -\alpha x_1 & \alpha x_1\\
                \alpha x_2 & -\alpha x_2
            \end{pmatrix} (\Delta_1(1), \Delta_1(2))=\E[u_1^{\msronePolicy}](1, 1)-(\msrstate{1}{1}, \msrstate{1}{2})\\
            \Delta_1(1)\bP\{u_1^{\msronePolicy}=w_1\}+\Delta_1(2)\bP\{u_1^{\msronePolicy}=w_2\}=0
        \end{cases}
        \end{equation*}
    Also, by the definition of CTMC, we have $$\begin{cases}
        \begin{pmatrix}
            -\alpha x_1 & \alpha x_1\\
            \alpha x_2 & -\alpha x_2
        \end{pmatrix}\begin{pmatrix}
            \bP\{u_1^{\msronePolicy}=w_1\} & \bP\{u_1^{\msronePolicy}=w_2\}
        \end{pmatrix} = 0
    \end{cases}.$$
    Hence, $$\begin{cases}
        \bP\{u_1^{\msronePolicy}=w_1\}=\frac{x_2}{x_1+x_2}\\
        \bP\{u_1^{\msronePolicy}=w_2\}=\frac{x_1}{x_1+x_2}
    \end{cases}.$$
    By definition of $v_1$, we have $$\begin{cases}
        \bP\{v_1^{\msronePolicy}=w_1\}=\frac{\msrstate{1}{1}}{\E[u_1^{\msronePolicy}]}\frac{x_2}{x_1+x_2}\\
        \bP\{v_1^{\msronePolicy}=w_2\}=\frac{\msrstate{1}{2}}{\E[u_1^{\msronePolicy}]}\frac{x_1}{x_1+x_2}
    \end{cases}.$$
    Solving this system, we get 
        $$\begin{cases}
            \Delta_1(1)=\alpha x_1 \frac{\msrstate{1}{1}-\msrstate{1}{2}}{(x_1+x_2)^2}\\
            \Delta_1(2)=\alpha x_2 \frac{\msrstate{1}{2}-\msrstate{1}{1}}{(x_1+x_2)^2}
        \end{cases}.$$
    Plugging into Theorem \ref{thm:msr-bound}, we get $$\E[Q_1^{\msrPolicy(\alpha)}]=\frac{\rho_1}{1-\rho_1}+\frac{x_1x_2(\msrstate{1}{1}-\msrstate{1}{2})^2}{\alpha\E[u_1^{\msronePolicy}](1-\rho_1)(x_1+x_2)^3}+B,$$
    where $x_2 \frac{\msrstate{1}{2}-\msrstate{1}{1}}{\alpha(x_1+x_2)^2} \le B\le   x_1 \frac{\msrstate{1}{1}-\msrstate{1}{2}}{\alpha(x_1+x_2)^2}$.
    Thus, $$\E[Q_1^{\msrPolicy(\alpha)}]\le \frac{\rho_1}{1-\rho_1}+\frac{x_1x_2(\msrstate{1}{1}-\msrstate{1}{2})^2}{\alpha\E[u_1^{\msronePolicy}](1-\rho_1)(x_1+x_2)^3}+x_1 \frac{\msrstate{1}{1}-\msrstate{1}{2}}{\alpha(x_1+x_2)^2}.$$
    Let $z=\frac{x_1x_2(\msrstate{1}{1}-\msrstate{1}{2})^2}{\alpha\E[u_1^{\msronePolicy}](1-\rho_1)(x_1+x_2)^3}+x_1 \frac{\msrstate{1}{1}-\msrstate{1}{2}}{\alpha(x_1+x_2)^2}$. We can see that $\lim_{\alpha\rightarrow\infty} z = 0$ and that $z\ge 0$. Analysis on $\E[Q_2^{\msrPolicy}]$ upper bound is analogous.
    
    Hence, we have shown that setting $\alpha=\infty$ minimizes this queue length upper bound.
\end{proof}

\section{Proof of Theorem \ref{thm:nmsr-existence}}
\label{app:nmsr}
\nmsrexistence*
\begin{proof}
    From Theorem \ref{thm:throughput-optimal}, we know that there exists an MSR policy $\msrPolicy$ that can stabilize the system.
    Hence, given some MSR policy $\msrPolicy$ that stabilizes the system, we will show how to construct a non-preemptive policy that also stabilizes the system.
    
    Obviously, if $\msrPolicy$ is already non-preemptive, we have proven our claim.
    Otherwise, we can use $\msrPolicy$ to construct an nMSR policy that stabilizes the system.
    
    To do this, we first define $p(\alpha)$ to be the MSR policy that with infinitesimal generator ${\bf G}^{\msrPolicy(\alpha)}=\alpha{\bf G}^{\msrPolicy}$.
    That is, $p(\alpha)$ is an MSR policy whose modulating process is a scaled version of $p$'s modulating process.
    We will create a non-preemptive policy, $\nmsrPolicy(\alpha)$, that is based on $\msrPolicy(\alpha)$. 
    To make $\nmsrPolicy(\alpha)$ non-preemptive, we must add some switching states between the states of $\msrPolicy(\alpha)$'s modulating process.
    The original states of $\msrPolicy(\alpha)$'s modulating process will form the working states of $\nmsrPolicy(\alpha)$.
    We assume that the added switching states form arbitrary (finite length) switching routes between these working states.
    Hence, $\nmsrPolicy(\alpha)$ is a valid nMSR policy.
    We will now show that, for sufficiently small $\alpha$, $\nmsrPolicy(\alpha)$ stabilizes the system.

    Let $w_1, w_2, \cdots, w_{N^{\msrPolicy}}$ denote the working states of $\nmsrPolicy(\alpha)$, and let $s_1, s_2, \cdots, s_{N^{\nmsrPolicy}-N^{\msrPolicy}}$ denote all the switching states of $\nmsrPolicy(\alpha)$. Let $\pi^{\nmsrPolicy(\alpha)}$ denote the stationary distribution of states in the modulating process of $\nmsrPolicy(\alpha)$.
    Let $\pi_{W}^{\nmsrPolicy(\alpha)}$ be the probability that we are in any working state defined as $$\pi_{W}^{\nmsrPolicy(\alpha)}=\sum_{i=1}^{N^{\msrPolicy}}\pi_{w_i}^{\nmsrPolicy(\alpha)}.$$

    By Lemma \ref{thm:existence}, we have $\bm{\mu}\otimes \E[\schedule^{\msrPolicy}] > \bm{\lambda}$.
    Therefore, there exists $\epsilon_1>0$ such that $\bm{\mu}\otimes \E[\schedule^{\msrPolicy}]-\epsilon_1=\bm{\lambda}$.
    Next, we condition $\E[\schedule^{\nmsrPolicy(\alpha)}]$ on whether it is in working state.
    \begin{align*}               \E[\schedule^{\nmsrPolicy(\alpha)}]&=\E[\schedule^{\nmsrPolicy(\alpha)}\mid \text{working state}] \pi_{W}^{\nmsrPolicy(\alpha)}+ \E[\schedule^{\nmsrPolicy(\alpha)}\mid \text{not working state}] (1-\pi_{W}^{\nmsrPolicy(\alpha)})\\
    &\ge \E[\schedule^{\nmsrPolicy(\alpha)}\mid \text{working state}] \pi_{W}^{\nmsrPolicy(\alpha)}\\
    &=\E[\schedule^{\msrPolicy(\alpha)}]\pi_{W}^{\nmsrPolicy(\alpha)}\\
    &=\E[\schedule^{\msrPolicy}]\pi_{W}^{\nmsrPolicy(\alpha)}
    \end{align*}
    Because Lemma \ref{thm:existence} tells us that 
    $\E[\schedule^{\msrPolicy}]-\epsilon_1=\bm{\lambda}$ for some $\epsilon_1 > 0$, it will suffice to show that 
     $\pi_W^{\nmsrPolicy(\alpha)}$ can be arbitrarily close to 1 when $\alpha$ is small. We prove this in Lemma \ref{lemma:epsilon-nmsr}.
    \begin{lemma}
        \label{lemma:epsilon-nmsr}
        For any $\epsilon>0$, there exists an $\alpha$ such that
        $$\pi_{W}^{\nmsrPolicy(\alpha)}> 1-\epsilon$$
    \end{lemma}
    \begin{proof}
        We will prove the Lemma via the renewal-reward theorem\cite{harchol2013performance}.
        We let the renewal cycle be the time between visits to state $w_1$ by $\schedule^{\nmsrPolicy(\alpha)}(t)$.
        Consider a path $\schedule^{\nmsrPolicy(\alpha)}(t)$ takes, $\mpath$, during a renewal cycle, we use the random variable $X$ to define the cycle length and let the reward function be $R^{\nmsrPolicy(\alpha)}(t)=\bm{1}[\schedule^{\nmsrPolicy(\alpha)}(t) \text{ is in a working state}]$. Let $R^{\nmsrPolicy(\alpha)}$ be the random variable denoting the reward accrued during one renewal cycle.
        Let $S^{\nmsrPolicy(\alpha)}$ be the random variable denoting the time $\schedule^{\nmsrPolicy(\alpha)}$ spends in switching states during a renewal cycle.

        Then, by the renewal-reward theorem \cite{harchol2013performance}, we have
        \begin{align}
            \pi_{W}^{\nmsrPolicy(\alpha)}&=\frac{\E[R^{\nmsrPolicy(\alpha)}]}{\E[X]}\nonumber\\
            &=\frac{\E[R^{\nmsrPolicy(\alpha)}]}{\E[R^{\nmsrPolicy(\alpha)}]+\E[S^{\nmsrPolicy(\alpha)}]}\nonumber \\
            &=\frac{\frac{1}{\alpha}\E[R^{\nmsrPolicy(1)}]}{\frac{1}{\alpha}\E[R^{\nmsrPolicy(1)}]+\E[S^{\nmsrPolicy(1)}]} \label{step:nmsr-scale}\\
            &=\frac{\E[R^{\nmsrPolicy(1)}]}{\E[R^{\nmsrPolicy(1)}]+\alpha\E[S^{\nmsrPolicy(1)}]}.\nonumber
        \end{align}
        Step (\ref{step:nmsr-scale}) follows from \begin{align*}
            \E[R^{\nmsrPolicy(\alpha)}]&=\sum_{\omega} \E[R^{\nmsrPolicy(\alpha)}\mid \omega] \bP\{\omega\}\\
            &=\sum_{\omega} \frac{1}{\alpha}\E[R^{\nmsrPolicy(1)}\mid \omega] \bP\{\omega\}\\
            &=\frac{1}{\alpha} \E[R^{\nmsrPolicy(1)}]
        \end{align*}
        because the transition rates out of the working states are being multiplied by $\alpha$ in $\nmsrPolicy(\alpha)$ compared to $\nmsrPolicy(1)$.
        Therefore, we only need to set the $\alpha$ to be $$0<\alpha<\frac{\epsilon \E[R^{\nmsrPolicy(1)}]}{(1-\epsilon)\E[S^{\nmsrPolicy(1)}]}$$
        such that $\pi_{W}^{\nmsrPolicy(\alpha)}>1-\epsilon$.
    \end{proof}
    From Lemma \ref{lemma:epsilon-nmsr}, we know that there exists some $\alpha$ such that $\pi_{W}^{\nmsrPolicy(\alpha)}=1-\epsilon_2$ holds for a given $\epsilon_2>0$.
    Hence, if we are scaling the transitions out of working states by $\alpha$, we have, \begin{align*}
        \bm{\mu}\otimes \E[\schedule^{\nmsrPolicy}]&=\bm{\mu}\otimes \E[\schedule^{\msrPolicy}]\pi_{W}(\alpha)\\
        &>\bm{\mu}\otimes \E[\schedule^{\msrPolicy}](1-\epsilon_2)\\
        &=(\bm{\lambda}+\epsilon_1)(1-\epsilon_2)\\
        &=\bm{\lambda}+\epsilon_1-\epsilon_2-\epsilon_1\epsilon_2.
    \end{align*}
    In this case, we can set $\epsilon_3=\epsilon_1-\epsilon_2-\epsilon_1\epsilon_2$. We need $\epsilon_3=\epsilon_1(1-\epsilon_2)-\epsilon_2>0$. Therefore, given $\epsilon_1$, we can set $\epsilon_2$ to be $\frac{\epsilon_2}{1-\epsilon_2}<\epsilon_1$.

    Lastly, by invoking Lemma \ref{thm:existence}, we have shown that the constructed nMSR policy with the $\alpha$ associated with $\epsilon_2$ can stabilize the system.
\end{proof}
\section{Example System from Section \ref{sec:preemptions}}
\begin{figure*}[b]
    \centering
    \begin{subfigure}[t]{1\linewidth}
        \centering
        \begin{tikzpicture}[->,>=stealth,line width=0.5pt,node distance=2.5cm]
           \node [circle,draw] (zero) {$w_{1,4,0}$};
           \node [circle,draw] (three)[right of=zero] {$w_{0,0,2}$};
           \path (zero) edge [bend left] node [above] {$\alpha{\bf G}_{w_{1,4,0},w_{0,0,2}}$} (three);
           \path (three) edge [bend left] node [below] {$\alpha{\bf G}_{w_{0,0,2}, w_{1,4,0}}$} (zero);
        \end{tikzpicture}
        \caption{The modulating process of the pMSR policy, $\msrPolicy(\alpha)$}
        \label{fig:mp:pmsr}
    \end{subfigure}
    \begin{subfigure}[t]{1\linewidth}
        \centering
        \begin{tikzpicture}[->,>=stealth,line width=0.5pt,node distance=1.8cm]
           \node [circle,draw] (zero) {$w_{1,4,0}$};
           \node [circle,draw] (half)[right of=zero] {$t_{1, 4, 0}$};
           \node [circle,draw] (one)[right of=half] {$t_{0, 4, 0}$};
           \node [circle,draw] (two)[right of=one] {$t_{0, 3, 0}$};
           \node [circle,draw] (four)[below of=two] {$t_{0,0,2}$};
           \node [circle,draw] (five)[left of=four] {$t_{0,0,1}$};
           \node [circle,draw] (six)[right of=two] {$t_{0,2,0}$};
           \node [circle,draw] (seven)[right of=six] {$t_{0,1,0}$};
           \node [circle,draw] (three)[right of=seven] {$w_{0,0,2}$};
           \path (zero) edge [bend left] node [above] {$\alpha{\bf G}_{w_{1,4,0},w_{0,0,2}}$} (half);
           \path (half) edge [bend left] node [above] {$\mu_1$} (one);
           \path (one) edge [bend left] node [above] {$4\mu_2$} (two);
           \path (two) edge [bend left] node [above] {$3\mu_2$} (six);
           \path (six) edge [bend left] node [above] {$2\mu_2$} (seven);
           \path (seven) edge [bend left] node [above] {$\mu_2$} (three);
           \path (three) edge [bend left] node [below right] {$\alpha{\bf G}_{w_{0,0,2}, w_{1,4,0}}$} (four);
           \path (four) edge [bend left] node [below] {$2\mu_3$} (five);
           \path (five) edge [bend left] node [below left] {$\mu_3$} (zero);
        \end{tikzpicture}
        \caption{The modulating process of the nMSR policy, $\nmsrPolicy(\alpha)$}
        \label{fig:mp:nmsr}
        \label{fig:nmsr-states}
    \end{subfigure}
    \begin{subfigure}[b]{1\linewidth}
        \centering
        \begin{tikzpicture}[->,>=stealth,line width=0.5pt,node distance=1.8cm]
           \node [circle,draw] (zero) {$w_{1,4,0}$};
           \node [circle,draw] (half)[right of=zero] {$t_5$};
           \node [circle,draw] (one)[right of=half] {$t_4$};
           \node [circle,draw] (two)[right of=one] {$t_3$};
           \node [circle,draw] (four)[below of=two] {$t_2'$};
           \node [circle,draw] (five)[left of=four] {$t_1'$};
           \node [circle,draw] (six)[right of=two] {$t_2$};
           \node [circle,draw] (seven)[right of=six] {$t_1$};
           \node [circle,draw] (three)[right of=seven] {$w_{0,0,2}$};
           \path (zero) edge [bend left] node [above] {$\alpha{\bf G}_{w_{1,4,0},w_{0,0,2}}$} (half);
           \path (half) edge [bend left] node [above] {$5\gamma$} (one);
           \path (one) edge [bend left] node [above] {$4\gamma$} (two);
           \path (two) edge [bend left] node [above] {$3\gamma$} (six);
           \path (six) edge [bend left] node [above] {$2\gamma$} (seven);
           \path (seven) edge [bend left] node [above] {$\gamma$} (three);
           \path (three) edge [bend left] node [below right] {$\alpha{\bf G}_{w_{0,0,2}, w_{1,4,0}}$} (four);
           \path (four) edge [bend left] node [below] {$2\gamma$} (five);
           \path (five) edge [bend left] node [below left] {$\gamma$} (zero);
        \end{tikzpicture}
        \caption{The modulating process of the sMSR policy, $\smsrPolicy(\alpha)$}
        \label{fig:mp:smsr}
        \label{fig:smsr-states}
    \end{subfigure}
    \caption{Here we depict the structure of the modulating processes for $\msrPolicy(\alpha)$, $\nmsrPolicy(\alpha)$ and $\smsrPolicy(\alpha)$. These policies are used for the examples described in Section \ref{sec:preemptions}.  States marked with $w$ are working states, and states marked with $t$ are switching states. In (a) and (b), the subscripts denote the state's corresponding schedule. In (c), the subscripts for switching states denote the number of jobs experiencing a setup time due to being preempted.}
\label{fig:mc}
\end{figure*}
\label{sec:example}
To evaluate our MSR policies in Section \ref{sec:preemptions}, we consider a server with three resource types (WLOG, we call the resources CPU cores, DRAM, and network bandwidth).
Our example server has 20 cores, 15 GB of DRAM and 50 Gbps of network bandwidth.
We assume there are three job types with demands 
$${\bf D}_1 = (3,7,1) \quad {\bf D}_2 = (4,1,1) \quad {\bf D}_3=(10,1,5).$$
When the system load, $\rho$, equals 1, we let 
$$\bm{\lambda} = (.5, 2, 1) \quad \bm{\mu}=(1, 1, 1).$$
When we vary the system load in our experiments, we keep $\bm{\mu}$ fixed, and scale ${\bm\lambda}$ linearly such that
$${\bm\lambda} = \rho \cdot (.5, 2, 1).$$
Note that the arrival rate vector was chosen such that $\rho < 1$ implies $\bm{\lambda} \in \mathcal{C}$.

This example was chosen to be non-trivial in two key ways.  
First, some possible schedules in this example are Pareto optimal, but do not lie on the boundary of the convex hull of $\mathcal{S}$.
For example, the schedule $(1,1,1)$ cannot fit an additional job of any type without removing some jobs from the schedule, but it does not lie on the convex boundary of $\mathcal{S}$.
Hence, although the schedule $(1,1,1)$ seems to pack the server efficiently, choosing this schedule will not help stabilize the system when system load is sufficiently high.

Second, the possible schedules on the convex boundary of $\mathcal{S}$ are not trivially close together.
Hence, in cases where preemption is limited, nMSR and sMSR policies will need non-trivial switching routes to move between working states.
The modulating processes for the pMSR, nMSR, and sMSR policies we construct based on this example are shown in Figure \ref{fig:mc}.
\section{Phase-Type Job Sizes}
\label{sec:appendix:phase-type}
In the modulating processes of pMSRs and sMSRs, the state transitions do not coincide with job completions.
Therefore, given jobs whose sizes follow a different distribution, the modulating process will not change.
As described in \cite{grosof2024analysis}, the RESET and MARC techniques used in the proof of Theorem \ref{thm:msr-bound} can be expanded to handle the case where job sizes follow phase-type distribution.
Hence, our queue length bounds from this theorem will continue to hold in this case.
Specifically, when dealing with phase-type distributions, 
the analysis of the MSR-1 system requires computing relative completions terms for each phase of type-$i$ jobs.
Hence, while the computation of our bounds is still straightforward in this case, there is some added complexity to having these additional job phases.

The story is more complicated in the nMSR case, where the modulating process relies on job completions to transition between working states.
Here, many additional states need to be added to the modulating process to capture all the possible switching routes that the nMSR policy may be forced to follow between working states.
Hence, the size of the modulating process might grow exponentially with the number of job phases.
Finding a compact representation of the modulating process for nMSR policies with phase-type jobs remains an open problem that we defer to future work.

\section{Proof of Theorem \ref{thm:ht}}
\label{sec:ht}
Our goal is to compare an MSR policy to MaxWeight in a heavy traffic scaling regime where $\rho \to 1$.
Due to the definition of $\rho$, there are many possible ways to achieve this scaling.
In Definition \ref{def:ht}, we define an intuitive heavy-traffic scaling regime where the arrival rate vector is scaled linearly to increase system load.

\begin{definition}\label{def:ht}
Consider a system with initial arrival rate vector $\bm{\lambda}$ and initial system load $\rho_0$.
Let $\bm{\tilde{\lambda}} = \frac{\bm{\lambda}}{\rho_0}$.
Here, $\bm{\tilde{\lambda}}$ represents the supremum of the set of scaled arrival rate vectors under which the system may be stable.
We consider the system in the limit as the arrival rate vector scales linearly to $\bm{\tilde{\lambda}}$.

Formally, for any value of system load, $\rho$, let the scaled arrival rate, $\bm{\lambda}(\rho)$, be defined as
$$\bm{\lambda}(\rho) = \bm{\tilde{\lambda}}\cdot \rho.$$
We consider the heavy-traffic limit as $\rho \to 1$.
\end{definition}

Using this definition of heavy traffic, we prove the following theorem.

\competitiveratio*
\begin{proof}

First, we prove that the solution to \eqref{eq:optimization} is the same at all loads.
This follows directly from Definition \ref{def:ht}, as \eqref{eq:optimization} is effectively finding a policy $p^*$ whose average completion rate is $\tilde{\bm\lambda}$.
Because $\tilde{\bm\lambda}$ is the same at all loads, the policy $p^*$ does not depend on $\rho$.

To prove that $p^*$ is constant competitive, we note that Lemma \ref{lemma:solve-delta} tells us $\E[\Delta(v_i^{p^{*}-1})]$ is only dependent on ${\bf G}$, the infinitesimal generator of of the modulating process.
Since the policy $p^*$ does not change with $\rho$, $\E[\Delta(v_i^{p^{*}-1})]$ is a constant.
Therefore, $\E[Q_i^{p^*-1}]$ scales as $\Theta(1/(1-\rho_i))$.
Our optimization objective in \eqref{eq:optimization} ensures that $\rho=\rho_i$ for all $i$, therefore $\E[Q^{p^*}]$ scales as $\Theta(1/(1-\rho))=\Theta(1/\epsilon)$. 
Hence, to prove that $p^*$ is constant-competitive with MaxWeight, it suffices to show that $\E[Q^{\textrm{MaxWeight}}]$ also scales as $\Theta(1/(1-\rho))$.

It is shown in \cite{eryilmaz2012asymptotically} that $\E[Q^{\textrm{MaxWeight}}]$ scales as $\Theta(1/\epsilon)$ in the limit as $\epsilon \to 0$, where $\epsilon$ is the minimum distance between ${\bm\lambda}$ and the faces of $Conv(\mathcal{S})$.
When $\epsilon$ is sufficiently small, the face of $Conv(\mathcal{S})$ closest to ${\bm\lambda}$ does not change as $\rho\to 1$.
Hence, for sufficiently large $\rho$, let $\theta$ be the angle between the vector ${\bm\lambda}$ and the normal vector of the closest face on $Conv(\mathcal{S})$.
Here, $\epsilon=(1-\rho)\cos(\theta)$.
Because $\epsilon$ and $1-\rho$ differ by at most a constant factor, $\E[Q^{\textrm{MaxWeight}}]$ scales as $\Theta(1/(1-\rho))$ as $\rho \to 1$.  Hence, $p^*$ is constant-competitive with MaxWeight.
\end{proof}
\section{Description of the Google Borg Trace}
\label{sec:trace}
Our trace is based on data collected from Google Borg over 30 days in 2019 \cite{tirmazi2020borg}.
The Borg dataset is divided into 8 cells that represent collections of machines in 8 different time zones.
The trace describes job arrival times, sizes, and the normalized CPU and memory demands of each job.

We considered data from one cell.
We first group the jobs in the trace into distinct job types, grouping together jobs whose CPU and memory demands are within 0.1\%.
While this resulted in hundreds of thousands of job types, the ten most frequent job types accounted for 51\% of all requests.
Hence, to construct a tractable scenario, we consider these ten most popular types.
We further filtered out jobs that were killed prematurely or did not finish during the observation window of the trace.

The trace of the top ten job types contains 69 million requests.  
The average job size across all types was 310 seconds.
The job size distribution is far from exponential -- the distribution's failure rate is initially increasing and then decreases quickly after 50 seconds.
The arrival process in the trace is highly variable and often includes large batches of arriving jobs.

\section{An Example on the Initialization Step of MSR}
\label{sec:init-example}
Consider a server with 128 cores of CPU, 256 GiB of memory, 1TiB of disk space, and 100Gbps of network.
Table \ref{table:example-init} lists the 4 types of virtual machines some cloud company wants to sell.
For simplicity, we assume the mean job sizes are all 1 second for all job types.
That is, ${\bm\mu}=(1,1,1,1)$.
The measurements at the company tell us we get 4, 5, 2, 1.5 virtual machine requests per second for type-1, type-2, type-3, type-4 virtual machines, respectively.
We can therefore plug these numbers into equation \eqref{eq:optimization} to solve for the candidate set and the fraction of time we need to spend in each candidate.
In this case, the capacity of the system is ${\bf P}=(128, 256, 1024, 100)$, the demand matrix is $${\bf D} = \begin{pmatrix}
    1 & 4 & 50 & 1\\
    4 & 1 & 10 & 10\\
    2 & 2 & 100 & 5\\
    8 & 4 & 10 & 1
\end{pmatrix},$$
and the arrival rate vector being ${\bm\lambda}=(4,5,2,1.5)$.

It takes the Gurobi solver under 1ms to solve for this MIQCP. The results are:
$${\bf W} = \begin{pmatrix}
    0 & 7 & 4 & 10\\
    0 & 5 & 9 & 5\\
    0 & 10 & 0 & 0\\
    10 & 9 & 0 & 0
\end{pmatrix}$$
and 
$${\bm \pi} = \begin{pmatrix}
    0.07633588 & 0.30534351 & 0.00763359 & 0.61068702
\end{pmatrix}.$$

\begin{table}[ht]
    \begin{tabular}{l|lllll}
      Type & CPU     & Memory & Disk   & Network & Demand\\
      \hline 
    Type-1 & 1 core  & 4GiB   & 50GiB  & 1Gbps & 4 requests/sec\\
    Type-2 & 4 cores & 1GiB   & 10GiB  & 10Gbps & 5 requests/sec\\
    Type-3 & 2 cores & 2GiB   & 100GiB & 5Gbps  & 2 requests/sec\\
    Type-4 & 8 cores & 4GiB   & 10GiB  & 1Gbps  & 1.5 requests/sec
    \end{tabular}
    \caption{The resource demands and arrival rates of the 4 job types in the example system.}
    \label{table:example-init}
\end{table}

\end{document}